\documentclass[12pt,a4paper]{amsart}

\usepackage{amsmath,amsfonts,amssymb}

\usepackage[hmargin=2cm,vmargin=2cm]{geometry}

\newcommand{\Z}{\mathbb{Z}}

\allowdisplaybreaks[3]

\newcommand{\fc}{\mathfrak{c}}

\newcommand{\ci}{\mathrm{i}}

\newcommand{\CC}{\mathbb{C}}
\newcommand{\ZZ}{\mathbb{Z}}

\newcommand{\DR}{\mathrm{DR}}
\newcommand{\cA}{\mathcal{A}}

\newcommand{\cC}{\mathcal{C}}

\newcommand{\cL}{\mathcal{L}}
\newcommand{\cM}{\mathcal{M}}

\newcommand{\sS}{\mathcal{S}}

\newcommand{\Ch}{\mathrm{Ch}}
\newcommand{\Td}{\mathrm{Td}}

\newcommand{\st}{\mathbf{H}}

\newcommand{\cvir}{c_{\mathrm{vir}}}
\newcommand{\degvir}{\mathrm{degvir}}

\newcommand{\mbZ}{\mathbb Z}
\newcommand{\mbC}{\mathbb C}

\newcommand{\mbQ}{\mathbb Q}
\def\CP1{\mathbb{C}\mathrm{P}^1}

\newcommand{\oM}{\overline{\mathcal M}}

\newcommand{\oh}{{\overline h}}
\newcommand{\og}{\overline g}

\newcommand{\eps}{\varepsilon}
\def\d{\partial}

\newcommand{\M}{\mathcal M}

\newcommand{\tu}{\widetilde u}

\newcommand{\hcA}{{\widehat{\mathcal A}}}

\newcommand{\hLambda}{{\widehat\Lambda}}

\renewcommand{\deg}{\mathop{\mathrm{deg}}\nolimits}
\newcommand{\res}{\mathop{\mathrm{res}}\nolimits}

\newcommand{\spec}{u^{\mathrm{sp}}}

\def\DR{{\rm DR}}

\newcommand{\even}{\mathrm{even}}
\newcommand{\ct}{\mathrm{ct}}
\newcommand{\GD}{\mathrm{GD}}
\newcommand{\str}{\mathrm{str}}

\newcommand{\wtop}{w^{\mathrm{top}}}
\newcommand{\fW}{\mathfrak{W}}
\newcommand{\tp}{\widetilde p}

\newtheorem{theorem}{Theorem}[section]
\newtheorem{thm}{Theorem}[section]
\newtheorem{proposition}[theorem]{Proposition}

\theoremstyle{definition}

\newtheorem{dfns}[theorem]{Definitions}

\newtheorem{rem}[theorem]{Remark}

\numberwithin{equation}{section}

\title[Double ramification hierarchy for Witten's $r$-spin class]{Towards a description of the double ramification hierarchy for Witten's $r$-spin class}

\author{Alexandr Buryak}
\address{A.~Buryak:\newline Department of Mathematics, ETH Zurich, \newline Ramistrasse 101 8092, HG G 27.1, Zurich, Switzerland}
\email{buryaksh@gmail.com}

\author{J\'er\'emy Gu\'er\'e}
\address{J.~Gu\'er\'e:\newline Humboldt Universit\"at,\newline Unter den Linden 6, 10099 Berlin, Germany}
\email{jeremy.guere@hu-berlin.de}

\begin{document}
\begin{abstract}
The double ramification hierarchy is a new integrable hierarchy of hamiltonian PDEs introduced recently by the first author. It is associated to an arbitrary given cohomological field theory. In this paper we study the double ramification hierarchy associated to the cohomological field theory formed by Witten's $r$-spin classes. Using the formula for the product of the top Chern class of the Hodge bundle with Witten's class, found by the second author, we present an effective method for a computation of the double ramification hierarchy. We do explicit computations for $r=3,4,5$ and prove that the double ramification hierarchy is Miura equivalent to the corresponding Dubrovin--Zhang hierarchy. As an application, this result together with a recent work of the first author with Paolo Rossi gives a quantization of the $r$-th Gelfand--Dickey hierarchy for $r=3,4,5$. 
\end{abstract}
\maketitle
\tableofcontents

\section{Introduction}

Integrable hierarchies play a major role in the study of cohomological field theories, such as Gromov--Witten theory of a projective variety or Fan--Jarvis--Ruan--Witten theory of an isolated singularity. A striking example is a conjecture of Witten~\cite{Wittenconj}, proved by Kontsevich~\cite{Kontsevich}, claiming that the partition function of the so-called trivial cohomological field theory (i.e.~the Gromov--Witten theory of a point) is a $\tau$-function of the KdV hierarchy. This gives recursion formulas to compute any intersection numbers involving psi-classes on the moduli space of stable algebraic curves~$\oM_{g,n}$.

In a recent paper \cite{Bur14}, a new integrable hierarchy of partial differential equations associated to a given cohomological field theory was introduced. The construction was inspired by ideas from Symplectic Field Theory~\cite{EGH00} and it makes use of the intersection numbers of the given cohomological field theory with the double ramification cycle, the top Chern class of the Hodge bundle and the psi-classes on the moduli space~$\oM_{g,n}$. 

In~\cite{Bur14} the author conjectured, guided by the examples of the trivial and the Hodge cohomological field theories (which give the KdV hierarchy and the hierarchy of the Intermediate Long Wave equation respectively) that the double ramification hierarchy is Miura equivalent to the Dubrovin--Zhang hierarchy associated to the same cohomological field theory via the construction described, for instance, in~\cite{DZ05} (see also~\cite{BPS12a,BPS12b}). In~\cite{BR14} the authors proved that the conjecture is true in the case of the cohomological field theory associated to the Gromov--Witten theory of the complex projective line.

In this paper we focus on the cohomological field theory formed by Witten's $r$-spin classes, the so-called $r$-spin theory. Originally, it was introduced in the context of the generalized Witten conjecture \cite{Witten2}, claiming that its partition function is a $\tau$-function of the $r$-th Gelfand--Dickey hierarchy. The $r$-spin theory was then developed by several authors \cite{D1,D2,D3,D4,D5,D6} and the generalized Witten conjecture was at last proved by Faber, Shadrin, and Zvonkine \cite{FSZ} after many other contributions \cite{C1,C2,C3,C4,C5,C6,C7,C8}.

Interestingly, the generalized Witten conjecture inspired the development of the quantum singularity theory by Fan, Jarvis, and Ruan \cite{FJRW,FJRW2}, also called Fan--Jarvis--Ruan--Witten (FJRW) theory. The $r$-spin theory corresponds to the singularity $x^r$ which is one particular example of the simple singularities:
\begin{equation*}
\begin{array}{ll}
x^{r+1},      & A_r-\textrm{case}, \\
x^2y+y^{r-1}, & D_r-\textrm{case}, \\
x^3+y^4,      & E_6-\textrm{case}, \\
x^3y+y^3,     & E_7-\textrm{case}, \\
x^3+y^5,      & E_8-\textrm{case}.
\end{array}
\end{equation*}
For each simple singularity there is a corresponding integrable hierarchy. These hierarchies are called the ADE-hierarchies. The $A_r$-hierarchy coincides with the $r$-th Gelfand--Dickey hierarchy. For simple singularities of type~$D$ and~$E$ these hierarchies were constructed by Drinfeld--Sokolov~\cite{DS} and later by Kac--Wakimoto \cite{KW} (these two constructions are equivalent by \cite{HM}). Fan, Jarvis, and Ruan~\cite{FJRW,FJRW2} generalized the result of~\cite{FSZ} proving that the partition functions of all simple singularities are $\tau$-functions of the corresponding integrable hierarchies.

Fan--Jarvis--Ruan--Witten theory is defined for any Landau--Ginzburg model $(W,G)$, that is the data of a quasi-homogeneous polynomial $W$ with an isolated singularity at the origin and of a group $G$ of diagonal matrices under which $W$ is invariant and which contains a matrix generated by the weights of the polynomial. Fan, Jarvis, and Ruan introduced a moduli space associated to the Landau--Ginzburg model~$(W,G)$ and constructed a virtual fundamental class on it, leading to a cohomological field theory.

FJRW theory is largely unknown in genus greater than zero and, in most cases, even the genus-zero invariants are hard to obtain. The situation is in fact very similar to Gromov--Witten theory; for instance, the main obstruction called non-concavity in FJRW theory is just as non-convexity in Gromov--Witten theory.

In a recent paper~\cite{Guere1}, the second author computed all genus-zero invariants of the quantum singularity theory of $(W,G_\mathrm{max})$ for chain polynomials 
$$
W = x_1^{a_1}x_2 + \ldots + x_{N-1}^{a_{N-1}}x_N + x_N^{a_N}
$$
and the maximal group $G_\mathrm{max}$. It involved to overcome non-concavity and it went through the Polishchuk--Vaintrob construction \cite{Polish1} of the virtual class by means of matrix factorizations. The second author explained how to get a two-periodic complex from these matrix factorizations and he introduced a new notion of a recursive complex to highlight their nice vanishing properties in cohomology. The main result is an explicit computation of the cohomology of a recursive complex, under some additional assumptions satisfied by the Polishchuk--Vaintrob construction in genus-zero for chain polynomials.

Interestingly, the same method provides some invariants for chain polynomials in an arbitrary genus as well. In~\cite{Guere2} (see also~\cite{PhDJG}), the second author proves an explicit formula for the cup product of the virtual class with the top Chern class of the Hodge bundle. Here we present this result in the case of the $r$-spin theory (see Theorem \ref{thmprodvirt}). Combining it with Chiodo's formula \cite{Chiodo1} and Hain's formula for the double ramification cycle~\cite{Hai11} we obtain an effective algorithm for a computation of the Hamiltonians of the double ramification hierarchy for the $r$-spin theory. In~\cite{BR14} the first author with P. Rossi found a simple recursion that allows to reconstruct the whole double ramification hierarchy starting from the Hamiltonian~$\og_{1,1}$. As a result, our method gives an effective way to reconstruct the double ramification hierarchy for the $r$-spin theory for an arbitrary fixed~$r$.

Another goal of the present paper is to compare the double ramification and the Dubrovin--Zhang hierarchies associated to the $r$-spin theory. It is known that, after certain rescalings, the Dubrovin--Zhang hierarchy coincides with the $r$-th Gelfand--Dickey hierarchy (\cite{Witten2,FSZ,DZ05}). In the $2$-spin case the associated double ramification hierarchy coincides with the Dubrovin--Zhang hierarchy that is the KdV hierarchy (see~\cite{Bur14}). In this paper, using our general method, we explicitly compute the Hamiltonian~$\og_{1,1}$ for the $3$, $4$ and $5$-spin theories and prove the following result. 
\begin{theorem}\label{theorem:main}
For the $3$-spin theory the double ramification hierarchy coincides with the Dubrovin--Zhang hierarchy. For the $4$ and $5$-spin theories the double ramification hierarchy is related to the Dubrovin--Zhang hierarchy by the following Miura transformation:
\begin{align*}
&\left\{
\begin{aligned}
&w^1=u^1+\frac{\eps^2}{96}u^3_{xx},\\
&w^2=u^2,\\
&w^3=u^3,
\end{aligned}\right.&&\text{for $r=4$};\\
&\left\{
\begin{aligned}
&w^1=u^1+\frac{\eps^2}{60}u^3_{xx},\\
&w^2=u^2+\frac{\eps^2}{60}u^4_{xx},\\
&w^3=u^3,\\
&w^4=u^4,
\end{aligned}\right.&&\text{for $r=5$}.
\end{align*}
\end{theorem}

The theorem has several important consequences. Since the Dubrovin--Zhang hierarchy for the $r$-spin theory is closely related to the $r$-th Gelfand--Dickey hierarchy, the recursion formulas from~\cite{BR14} give recursion formulas for the Hamiltonians of the $r$-th Gelfand--Dickey hierarchy for $r=3,4,5$. As far as we know, these recursion formulas never appeared in the literature before. Another application of Theorem~\ref{theorem:main} comes from the work~\cite{BR15} where the first author together with P. Rossi constructed a natural quantization of the double ramification hierarchy. Using this construction and Theorem~\ref{theorem:main} we obtain a quantization of the $r$-th Gelfand--Dickey hierarchy for $r=3,4,5$. As far as we know, this result is also new.    

We would like to say a few words about the way we prove Theorem~\ref{theorem:main}. Both the double ramification hierarchy and the Dubrovin--Zhang hierarchy consist of an infinite number of Hamiltonians and it seems to be difficult to prove directly that they are related by some Miura transformation. As we already said, there are simple recursions (see~\cite{BR14}) that allow to reconstruct the whole double ramification hierarchy starting from just one Hamiltonian~$\og_{1,1}$. The problem is that on the Dubrovin--Zhang side we don't know such recursions. Moreover, even in the concrete example of the $r$-spin theory it seems to be difficult to prove that the associated Dubrovin--Zhang hierarchy, after some Miura transformation, satisfies the same recursions as were found in~\cite{BR14}. In order to overcome this difficulty we obtain a general result that, we believe, can have an independent interest. We prove that a hamiltonian hierarchy of certain type can be uniquely reconstructed from just one Hamiltonian using also the string and the dilaton equations for a specific solution (see Proposition~\ref{proposition:reconstruction}). This result together with an explicit computation of the Hamiltonian $\og_{1,1}$ of the double ramification hierarchy and the known description of the Dubrovin--Zhang hierarchy allows us to prove Theorem~\ref{theorem:main}.    

\begin{rem}
As we were informed by B.~Dubrovin and S.~Shadrin, Proposition~\ref{proposition:reconstruction} is known to experts, but it seems that it didn't appear in the literature before. 
\end{rem}

\subsection{Organization of the paper}

In Section~\ref{section:DR hierarchy} we recall the construction of the double ramification hierarchy. 

In Section~\ref{section:DR hierarchy for r-spin} we review the construction of Witten's $r$-spin class and present the formula from~\cite{Guere2,PhDJG} for its product with the the top Chern class of the Hodge bundle. We show how to apply it to a computation of the Hamiltonians of the double ramification hierarchy for the~$r$-spin theory and compute explicitly the Hamiltonian~$\og_{1,1}$ for the $3,4$ and $5$-spin theory.

In Section~\ref{section:DZ hierarchy for r-spin} we recall the construction of the $r$-th Gelfand--Dickey hierarchy and its relation to the Dubrovin--Zhang hierarchy for the $r$-spin theory.

In Section~\ref{section:proof} we prove that the string solution of the double ramification hierarchy satisfies the dilaton equation. Then we prove Proposition~\ref{proposition:reconstruction} and Theorem~\ref{theorem:main}.

In Section~\ref{section:quantization} we obtain a quantization of the $r$-spin Dubrovin--Zhang hierarchy for $r=3,4,5$.

\subsection{Acknowledgments}

We are grateful to A.~Chiodo who organized the conference ``Mirror symmetry and spin curves'' in Cortona, the present work was started there.
We also thank B.~Dubrovin, P.~Rossi, S.~Shadrin, R.~Pandharipande and D.~Zvonkine for discussions related to the work presented here.

The first author was supported by grant ERC-2012-AdG-320368-MCSK in the group of R. Pandharipande at ETH Zurich, by grants RFFI 13-01-00755 and NSh-4850.2012.1. The second author was supported by the Einstein Stiftung.


\section{Double ramification hierarchy}\label{section:DR hierarchy}

In this section we briefly recall the main definitions from~\cite{Bur14} (see also~\cite{BR14}). The double ramification hierarchy is a system of commuting Hamiltonians on an infinite dimensional phase space that can be heuristically thought of as the loop space of a fixed vector space. The entry datum for this construction is a cohomological field theory  in the sense of Kontsevich and Manin~\cite{KM94}. Denote by $c_{g,n}\colon V^{\otimes n} \to H^{\even}(\oM_{g,n};\mbC)$ the system of linear maps defining the cohomological field theory, $V$ its underlying $N$-dimensional vector space, $\eta$ its metric tensor and $e_1\in V$ the unit of the cohomological field theory.

\subsection{The formal loop space} 

The loop space of $V$ will be defined somewhat formally by describing its ring of functions. Following \cite{DZ05} (see also \cite{Ros09}), let us consider formal variables~$u^\alpha_i$, $\alpha=1,\ldots,N$, $i=0,1,\ldots$, associated to a basis $e_1,\ldots,e_N$ of $V$. Always just at a heuristic level, the variable $u^\alpha:=u^\alpha_0$ can be thought of as the component $u^\alpha(x)$ along $e_\alpha$ of a formal loop $u\colon S^1\to V$, where $x$ is the coordinate on $S^1$, and the variables $u^\alpha_{x}:=u^\alpha_1, u^\alpha_{xx}:=u^\alpha_2,\ldots$ as its $x$-derivatives. We then define the ring $\cA_N$ of {\it differential polynomials} as the ring of polynomials $f(u;u_x,u_{xx},\ldots)$ in the variables~$u^\alpha_i, i>0$, with coefficients in the ring of formal power series in the variables $u^\alpha=u^\alpha_0$. We can differentiate a differential polynomial with respect to $x$ by applying the operator $\partial_x := \sum_{i\geq 0} u^\alpha_{i+1} \frac{\partial}{\partial u^\alpha_i}$ (in general, we use the convention of sum over repeated greek indices, but not over repeated latin indices). Finally, we consider the quotient~$\Lambda_N$ of the ring of differential polynomials first by constants and then by the image of~$\partial_x$, and we call its elements {\it local functionals}. A local functional, that is the equivalence class of a differential polynomial~$f=f(u;u_x,u_{xx},\ldots)$, will be denoted by $\overline{f}=\int f dx$. 

Differential polynomials and local functionals can also be decribed using another set of formal variables, corresponding heuristically to the Fourier components $p^\alpha_k$, $k\in\mbZ$, of the functions $u^\alpha=u^\alpha(x)$. Let us, hence, define a change of variables
\begin{gather}\label{eq:u-p change}
u^\alpha_j = \sum_{k\in\mbZ} (i k)^j p^\alpha_k e^{i k x},
\end{gather}
which allows us to express a differential polynomial $f(u;u_x,u_{xx},\ldots)$ as a formal Fourier series in $x$ where the coefficient of $e^{i k x}$ is a power series in the variables $p^\alpha_j$ (where the sum of the subscripts in each monomial in $p^\alpha_j$ equals $k$). Moreover, the local functional~$\overline{f}$ corresponds to the constant term of the Fourier series of $f$.

Let us describe a natural class of Poisson brackets on the space of local functionals. Given an $N\times N$ matrix~$K=(K^{\mu\nu})$ of differential operators of the form $K^{\mu\nu} = \sum_{j\geq 0} K^{\mu\nu}_j \partial_x^j$, where the coefficients $K^{\mu\nu}_j$ are differential polynomials and the sum is finite, we define
$$
\{\overline{f},\overline{g}\}_{K}:=\int\left(\frac{\delta \overline{f}}{\delta u^\mu}K^{\mu \nu}\frac{\delta \overline{g}}{\delta u^\nu}\right)dx,
$$
where we have used the variational derivative $\frac{\delta \overline{f}}{\delta u^\mu}:=\sum_{i\geq 0} (-\partial_x)^i \frac{\partial f}{\partial u^\mu_i}$. Imposing that such bracket satisfies the anti-symmetry and the Jacobi identity will translate, of course, into conditions for the coefficients~$K^{\mu \nu}_j$. An operator that satisfies such conditions will be called hamiltonian. A standard example of a hamiltonian operator is given by $\eta \partial_x$. The corresponding Poisson bracket also has a nice expression in terms of the variables $p^\alpha_k$:
$$
\{p^\alpha_k, p^\beta_j\}_{\eta \partial_x} = i k \eta^{\alpha \beta} \delta_{k+j,0}.
$$

Finally, we will need to consider extensions $\hcA_N$ and $\hLambda_N$ of the spaces of differential polynomials and local functionals. First, let us introduce a grading $\deg u^\alpha_i = i$  and a new variable~$\eps$ with $\deg\eps = -1$. Then $\hcA^{[k]}_N$ and $\hLambda^{[k]}_N$ are defined, respectively, as the subspaces of degree~$k$ of $\hcA_N:=\cA_N\otimes\mbC[[\eps]]$ and of~$\hLambda_N:=\Lambda_N\otimes\mbC[[\eps]]$. Their elements will still be called differential polynomials and local functionals. We can also define Poisson brackets as above, starting from a hamiltonian operator $K=(K^{\mu\nu})$, $K^{\mu\nu} = \sum_{i,j\geq 0} K^{\mu\nu}_{ij} \eps^i \partial_x^j$, where $K^{\mu\nu}_{ij}$ are differential polynomials of degree~$i-j+1$. The corresponding Poisson bracket will then have degree $1$. 

Note that to any local functional~$\oh=\sum_{i\ge 0}\oh_i$, $\deg\oh_i=i$, from~$\Lambda_N$ we can naturally associate a local functional~$\oh'$ from~$\hLambda_N^{[0]}$ by $\oh':=\sum_{i\ge 0}\eps^i\oh_i$. We can do the same procedure with any hamiltonian operator $K=(K^{\mu\nu})$, $K^{\mu\nu}=\sum_{i\ge 0}K^{\mu\nu}_i\d_x^i$, $K^{\mu\nu}_i\in\cA_N$. Let $K^{\mu\nu}_i=\sum_{j\ge 0}K^{\mu\nu}_{ij}$, where $\deg K^{\mu\nu}_{ij}=j$. From the anti-symmetry property of the bracket~$\{\cdot,\cdot\}_K$ it follows that~$K^{\mu\nu}_{00}=0$. Then, we define an operator $K'=((K')^{\mu\nu})$ by $(K')^{\mu\nu}=\sum_{i,j\ge 0}K^{\mu\nu}_{ij}\eps^{i+j-1}\d_x^i$. Therefore, we can naturally consider the space~$\Lambda_N$ as a subspace of~$\hLambda^{[0]}_N$ and the set of hamiltonian operators for~$\Lambda_N$ as a subset of hamiltonian operators for~$\hLambda_N$. 

A hamiltonian system of PDEs is a system of the form
\begin{gather}\label{eq:Hamiltonian system}
\frac{\partial u^\alpha}{\partial \tau_i} = K^{\alpha\mu} \frac{\delta\overline{h}_i}{\delta u^\mu}, \ \alpha=1,\ldots,N ,\ i=1,2,\ldots,
\end{gather}
where $\oh_i\in\hLambda^{[0]}_N$ are local functionals with the compatibility condition $\{\oh_i,\oh_j\}_K=0$, for $i,j\geq 1$. The local functionals~$\oh_i$ are called the {\it Hamiltonians} of the system~\eqref{eq:Hamiltonian system}.

\subsection{The double ramification hierarchy} 

Consider an arbitrary cohomological field theory $c_{g,n}\colon V^{\otimes n} \to H^{\even}(\oM_{g,n};\mbC)$. As usual, we denote by $\psi_i$ the first Chern class of the line bundle over~$\oM_{g,n}$ formed by the cotangent lines at the $i$-th marked point. Denote by~$\mathbb E$ the rank~$g$ Hodge vector bundle over~$\oM_{g,n}$ whose fibers are the spaces of holomorphic one-forms. Let $\lambda_j:=c_j(\mathbb E)\in H^{2j}(\oM_{g,n};\mbC)$. The Hamiltonians of the double ramification hierarchy are defined as follows:
\begin{gather}\label{DR Hamiltonians}
\og_{\alpha,d}:=\sum_{\substack{g\ge 0,n\ge 2}}\frac{(-\eps^2)^g}{n!}\sum_{\substack{a_1,\ldots,a_n\in\mbZ\\\sum a_i=0}}\left(\int_{\DR_g(0,a_1,\ldots,a_n)}\lambda_g\psi_1^d c_{g,n+1}\left(e_\alpha\otimes \bigotimes_{i=1}^n e_{\alpha_i}\right)\right)\prod_{i=1}^n p^{\alpha_i}_{a_i},
\end{gather}
for $\alpha=1,\ldots,N$ and $d=0,1,2,\ldots$. Here $\DR_g(a_1,\ldots,a_n) \in H^{2g}(\oM_{g,n};\mbQ)$ is the double ramification cycle. On~$\cM_{g,n}$ it can be defined as the Poincar\'e dual to the locus of pointed smooth curves~$[C,p_1,\ldots,p_n]$ satisfying
$$
\mathcal O_C\left(\sum_{i=1}^n a_ip_i\right)\cong\mathcal O_C,
$$ 
and we refer the reader, for example, to~\cite{Bur14} for the definition of the double ramification cycle on the whole moduli space~$\oM_{g,n}$.

The expression on the right-hand side of~\eqref{DR Hamiltonians} can be uniquely written as a local functional from $\hLambda_N^{[0]}$ using the change of variables~\eqref{eq:u-p change}. Concretely it can be done in the following way. The integral
\begin{gather}\label{DR integral}
\int_{\DR_g\left(0,a_1,\ldots,a_n\right)}\lambda_g\psi_1^d c_{g,n+1}\left(e_\alpha\otimes \bigotimes_{i=1}^n e_{\alpha_i}\right)
\end{gather}
is a polynomial in $a_1,\ldots,a_n$ homogeneous of degree~$2g$. It follows from Hain's formula~\cite{Hai11}, the result of~\cite{MW13} and the fact that $\lambda_g$ vanishes on $\oM_{g,n}\setminus\M_{g,n}^{\ct}$, where $\M_{g,n}^{\ct}$ is the moduli space of stable curves of compact type. Thus, the integral~\eqref{DR integral} can be written as a polynomial
\begin{gather*}
P_{\alpha,d,g;\alpha_1,\ldots,\alpha_n}(a_1,\ldots,a_n)=\sum_{\substack{b_1,\ldots,b_n\ge 0\\b_1+\ldots+b_n=2g}}P_{\alpha,d,g;\alpha_1,\ldots,\alpha_n}^{b_1,\ldots,b_n}a_1^{b_1}\ldots a_n^{b_n}.
\end{gather*}
Then we have
$$
\og_{\alpha,d}=\int\sum_{g\ge 0,n\ge 2}\frac{\eps^{2g}}{n!}\sum_{\substack{b_1,\ldots,b_n\ge 0\\b_1+\ldots+b_n=2g}}P_{\alpha,d,g;\alpha_1,\ldots,\alpha_n}^{b_1,\ldots,b_n} u^{\alpha_1}_{b_1}\ldots u^{\alpha_n}_{b_n}dx.
$$
Note that the integral~\eqref{DR integral} is defined only when $a_1+\ldots+a_n=0$. Therefore, a polynomial~$P_{\alpha,d,g;\alpha_1,\ldots,\alpha_n}$ is actually not unique. However, the resulting local functional $\og_{\alpha,d}\in\hLambda_N^{[0]}$ doesn't depend on this ambiguity (see~\cite{Bur14}). 

The fact that the local functionals~$\og_{\alpha,d}$ mutually commute with respect to the standard bracket~$\eta\d_x$ was proved in~\cite{Bur14}. The system of local functionals $\og_{\alpha,d}$, for $\alpha=1,\ldots,N$, $d=0,1,2,\ldots$, and the corresponding system of hamiltonian PDEs with respect to the standard Poisson bracket~$\{\cdot,\cdot\}_{\eta\partial_x}$,
$$
\frac{\d u^\alpha}{\d t^\beta_q}=\eta^{\alpha\mu}\d_x\frac{\delta\og_{\beta,q}}{\delta u^\mu},
$$
is called the \emph{double ramification hierarchy}.


\section{Computation of the double ramification hierarchy for the $r$-spin theory}\label{section:DR hierarchy for r-spin}

In this section, we first recall the construction of the $r$-spin cohomological field theory. Then we present the formula from~\cite{Guere2,PhDJG} for the cup product of Witten's class with the top Chern class of the Hodge bundle, see Theorem \ref{thmprodvirt}. Finally, we show how to apply this result to a computation of the Hamiltonians of the double ramification hierarchy and compute explicitly the Hamiltonian $\og_{1,1}$ for the $r$-spin theory, with $r \leq 5$.

\subsection{$r$-spin theory}

The $r$-spin theory is the cohomological field theory corresponding to the quantum singularity theory, or Fan--Jarvis--Ruan--Witten theory, of the polynomial $x^r$. The state space of this theory is
\begin{equation*}
\st := \bigoplus_{k=1}^{r-1} \CC \cdot e_k
\end{equation*}
with the pairing $(e_k,e_l)=\delta_{k+l,r}$ and with the grading
\begin{equation*}
\deg (e_k) = 2\cfrac{k-1}{r}.
\end{equation*}
The linear maps $c_{g,n} \colon \st ^{\otimes n} \rightarrow H^{\even}(\overline{\mathcal{M}}_{g,n};\CC)$ defining the cohomological field theory factorize through the cohomology of another moduli space, called the moduli space of $r$-spin curves. Let us briefly recall it.

\begin{dfns}
A genus-$g$ orbifold (or twisted) curve $\cC$ with marked points is a connected, proper, one-dimensional Deligne--Mumford stack whose coarse space $C$ is a genus-$g$ nodal curve, such that the morphism $\rho \colon \cC \rightarrow C$ is an isomorphism away from the nodes and the marked points. Any marked point or node has a non-trivial stabilizer equal to a finite cyclic group. An orbifold curve is called smoothable if its local picture at the node is $ \left\lbrace xy=0 \right\rbrace/\mathbb{U}(k)$ for some $k \in \ZZ$, where the action of the group $\mathbb{U}(k)$ of $k$-th roots of unity is defined by
\begin{equation*}
\zeta_k \cdot (x,y) := (\zeta_k x , \zeta_k^{-1} y), \quad \zeta_k:=e^{\frac{2 \pi \ci}{k}}.
\end{equation*}
An $r$-stable curve is a smoothable orbifold curve whose stabilizers (at the nodes and at the markings) have the same fixed order $r$ and whose coarse nodal pointed curve is stable. An $r$-spin curve of genus $g$ is the data
\begin{equation*}
(\cC; \sigma_1,\dotsc,\sigma_n;\cL;\phi)
\end{equation*}
of an $r$-stable genus-$g$ orbifold curve $\cC$ with marked points $\sigma_1,\dotsc,\sigma_n$ and of a line bundle $\cL$ on the curve $\cC$ satisfying the condition
\begin{equation}\label{spinstruct}
\phi \colon \cL^{\otimes r} \simeq \omega_{\cC,\log} := \omega_{\cC}(\sigma_1+\dotsc+\sigma_n),
\end{equation}
where $\omega_\cC$ is the canonical line bundle on~$\cC$. The moduli space of $r$-spin curves is the stack classifying all $r$-spin curves of genus $g$ with $n$ marked points. We denote it by~$\sS^r_{g,n}$.
\end{dfns}

The moduli space $\sS^r_{g,n}$ is a smooth and proper Deligne--Mumford stack of complex dimension~$3g-3+n$ and it is a finite cover of the moduli space $\overline{\cM}_{g,n}$ of stable curves. The projection $\textrm{o} \colon \sS^r_{g,n} \to \overline{\cM}_{g,n}$ is obtained by forgetting the line bundle and the stack structure.

Locally at the marked point $\sigma_i$ of an $r$-spin curve, the group $\mathbb{U}(r)$ of $r$-th roots of unity acts on the line bundle $\cL$ as 
\begin{equation}\label{multiplicities}
\zeta_r \cdot (x,\xi) = (\zeta_r x, \zeta_r^{m_i} \xi),
\end{equation}
where $m_i \in \left\lbrace 0 , \dotsc, r-1 \right\rbrace$ is called the multiplicity of the line bundle~$\cL$ at the marked point~$\sigma_i$. Hence we have a decomposition
\begin{equation*}
\sS^r_{g,n} = \bigsqcup_{m_i\in \left\lbrace 0,1,\dotsc,r-1\right\rbrace} \sS^r_{g,n}(m_1,\dotsc,m_n)
\end{equation*}
of the moduli space of $r$-spin curves into the moduli spaces $\sS^r_{g,n}(m_1,\dotsc,m_n)$ of $r$-spin curves with fixed multiplicities $(m_1,\dotsc,m_n)$ at the marked points.
For the space $\sS^r_{g,n}(m_1,\dotsc,m_n)$ to be non-empty, the following selection rule \cite[Proposition 2.2.8]{FJRW} must hold:
\begin{gather*}
\zeta_r^{m_1} \dotsm \zeta_r^{m_n} = \zeta_r^{2g-2+n},
\end{gather*}
or, equivalently, 
\begin{gather}\label{selecrule}
m_1+\ldots+m_n=2g-2+n\,\,(\text{mod $r$}).
\end{gather}

\begin{rem}
There is an alternative description of the moduli space of $r$-spin curves with multiplicities $(m_1,\dotsc,m_n)$, where a curve $\cC$ has no orbifold structure at the marked points (but one still has orbifold structures at the nodes) and where the condition \eqref{spinstruct} on the line bundle $\cL$ is replaced by
\begin{equation*}
\phi \colon \cL^{\otimes r} \simeq \omega_{\cC}\Bigl( \sum_{i=1}^n (1-m_i) \sigma_i \Bigr) .
\end{equation*}
Both definitions lead to the same moduli space and the same theory.
\end{rem}

Each component $\sS^r_{g,n}(m_1,\dotsc,m_n)$ possesses an important homology cycle whose Poincar\'e dual cohomology class has degree
\begin{equation*}
2\cdot \degvir = \deg (e_{m_1}) + \dotsb + \deg (e_{m_n}) + 2 \Bigl(1-\frac{2}{r}\Bigr) (g-1).
\end{equation*}
This cohomology class is called Witten's virtual class
\begin{equation*}
\cvir(m_1,\dotsc,m_n)_{g,n}\in H^{2\cdot \degvir} (\sS^r_{g,n}(m_1,\dotsc,m_n);\CC)
\end{equation*}
and it is the main ingredient of the $r$-spin cohomological field theory. Witten's virtual class has the property to vanish when at least one multiplicity $m_i$ is zero, and we define a linear map~$c_{g,n}^{\text{$r$-spin}}\colon\st^{\otimes n}\to H^{\even}(\oM_{g,n};\CC)$ by
\begin{gather}\label{dfnCohFT}
c^{\text{$r$-spin}}_{g,n}(e_{m_1}\otimes\ldots\otimes e_{m_n}):=(-1)^\mathrm{degvir} ~ r^{1-g} ~ \textrm{o}_* \cvir(m_1,\dotsc,m_n)_{g,n},\quad m_i\in\{1,2,\ldots,r-1\}.
\end{gather}
The rescaling coefficient~$(-1)^\mathrm{degvir} ~ r^{1-g}$ occurs since we have to divide by the degree~$r^{2g-1}$ of the projection~$\textrm{o}$ and to multiply by the order of the group~$\mathbb{U}_r$. We explain the sign~$(-1)^{\degvir}$ in the remark below. The collection of maps~$c_{g,n}^{\text{$r$-spin}}$ has all the properties of cohomological field theory and is called the {\it $r$-spin theory}.

\begin{rem}
There are two different conventions on the sign of the virtual class. In \cite{ChiodoJAG,LG/CY,Guere1}, the virtual class in the concave situation is given by the top Chern class of the vector bundle~$R^1\pi_*\cL$ (see equation \eqref{concave}), whereas in \cite{FJRW,Polish1} they choose the dual of this vector bundle. Even without concavity, the two conventions only differ by the sign $(-1)^{\textrm{degvir}}$, where the integer $\degvir$ is the half of the cohomological degree of the virtual class. We follow in this paper the convention of \cite{ChiodoJAG,LG/CY,Guere1} for the virtual class, as we find it more natural. Nevertheless, when dealing with the cohomological field theory, we have to incorporate this sign, as shown in equation \eqref{dfnCohFT}, for the potential of the $r$-spin theory to be a tau-function of the $r$-th Gelfand--Dickey hierarchy.
\end{rem}

\begin{rem}
The quantum singularity theory is defined for more general polynomial singularities $W$ and there are two constructions of the virtual class. One uses analytic methods and has been provided by Fan, Jarvis, and Ruan \cite{FJRW,FJRW2}. The other construction, by Polishchuk and Vaintrob \cite{Polish1}, is algebraic and uses a general set-up of matrix factorizations. It is not known in general whether these two constructions give the same cohomological class, but it is proved to be true for $W=x^r$ (see \cite[Theorem 1.2]{Li2}, or the more general result~\cite[Theorem 3.25]{Guere1}).
\end{rem}

\subsection{Hodge integrals for the $r$-spin theory}

The goal of this section is to present the formula from~\cite{Guere2,PhDJG} for the product
\begin{gather}\label{virtprod}
\lambda_g\cvir(m_1,\ldots,m_n)_{g,n},
\end{gather}
where, abusing our notations a little bit, we denote by the same letter~$\lambda_g$ the top Chern class of the Hodge bundle over~$\sS^r_{g,n}(m_1,\dotsc,m_n)$. Actually, the result from~\cite{Guere2,PhDJG} works in a much more general situation. Here we apply it to the particular case of the $r$-spin theory.

Although we are not going to define the Witten's $r$-spin class $c_{vir}$, we are going to discuss its product $\eqref{virtprod}$ with the class $\lambda_g$. Indeed, only this product will be used in the rest of the paper.
As a consequence, Theorem \ref{thmprodvirt} can be seen here as a definition, compatible with the original definitions \cite{FJRW,FJRW2} and \cite{Polish1,ChiodoJAG} of Witten's $r$-spin class.

Consider a family of orbifold curves $\pi \colon \cC \rightarrow S$ over a smooth and proper base $S$, together with a universal line bundle $\cL$ satisfying the algebraic relation
\begin{equation}\label{algrel}
\cL^{\otimes r} \simeq \omega_{\cC/S}(\sigma_1+\dotsb+\sigma_n),
\end{equation}
where $\omega_{\cC/S}= \omega_\pi$ is the relative canonical bundle of the morphism $\pi \colon \cC \rightarrow S$.
First, there is an ideal situation where the push-forward~$R^0\pi_*\cL$ vanishes and where the sheaf~$R^1\pi_*\cL$ is a vector bundle; it is the so-called concave situation and it happens exactly when the genus of the curves is zero.
In this case, we define the virtual class as the top Chern class of the vector bundle $R^1\pi_*\cL$, that is
\begin{equation}\label{concave}
\cvir := c_\textrm{top}(R^1\pi_*\cL) \quad \textrm{in genus zero}.
\end{equation}

Without concavity, we have to deal with the higher push-forward $R^\bullet\pi_*\cL$, that we represent as a complex of two vector bundles:
\begin{equation*}
R^\bullet\pi_*\cL = \left[ A \xrightarrow{d} B \right].
\end{equation*}
Over a geometric point $s \in S$, the kernel of the map $d$ is the vector space $H^0(\cC_s,\cL_s)$ and its cokernel is $H^1(\cC_s,\cL_s)$.
There is no natural extension of the top Chern class to a general K-theoretic element such as $R^\bullet\pi_*\cL$. Indeed, to respect multiplicativity of the top Chern class, we would need such an extension to be invertible and in particular to be possibly of negative cohomological degree.

Remarkably, the algebraic relation \eqref{algrel} gives us the opportunity to extend the definition of the top Chern class to the $K$-theoretic element 
$$
B+\mathbb{E}^\vee-A,
$$
where $\mathbb{E}$ is the Hodge bundle over~$S$ with fiber $H^0(\cC_s,\omega_{\cC_s})$ over a geometric point $s \in S$. More precisely, we use Polishchuk--Vaintrob's construction \cite{Polish1}, revisited as the cohomology of a recursive complex as in \cite{Guere1} and we end up with a specific characteristic class, that we now describe.

For a vector bundle $V$ on $S$ and a parameter $t \in \CC$, let us define the class
\begin{equation}\label{charclass}
\fc_t(V) := \Ch (\lambda_{-t} V^\vee)\Td (V) \in H^*(S)\left[t\right].
\end{equation} 
Here $\Ch$ denotes the Chern character, $\Td$ is the Todd class and~$\lambda_{-t}$ denotes the $\lambda$-ring structure of $K$-theory according to \cite[Ch V]{Fult}, that is
\begin{equation*}
\lambda_{t} V := \sum_{k \geq 0} (\Lambda^k V)t^k \in K^0(S)\left[t\right] .
\end{equation*}
By \cite[Ch.~I, Prop.~5.3]{Fult}, we have
\begin{equation*}
\lim_{t \to 1} \fc_t(V) = c_\textrm{top}(V).
\end{equation*}
The classes $\lambda_t V$ and $\fc_t(V)$ are invertible in $K^0(S)[\![ t ]\!]$ and $H^*(S)[\![ t ]\!]$ respectively. Therefore, the function $\fc_t$ can be defined for an arbitrary element from $K^0(S)$ as follows:
\begin{equation}\label{charclass2}
\fc_t(B-A) := \frac{\fc_t(B)}{\fc_t(A)} = \Ch \biggl( \frac{\lambda_{-t}B^\vee}{\lambda_{-t}A^\vee} \biggr) ~ \cfrac{\Td B}{\Td A} \in H^*(S)[\![ t ]\!],
\end{equation}
where $A$ and $B$ are two vector bundles. Since the class $c_\textrm{top}$ is not invertible, in general the radius of convergence of~$\fc_t$ as a function of~$t$ is equal to~$1$ and the limit~$t\to 1$ doesn't exist.

In~\cite{Guere1} it was observed that the characteristic class~$\fc_t$ can be extended to every $t \neq 1$ by taking
\begin{equation}\label{cvirx}
\fc_t(B-A) = \exp\left(\sum\limits_{l \geq 0} s_l(t) \Ch_l(A-B)\right), 
\end{equation}
with the functions
\begin{equation}\label{parametresl}
s_l(t) = \left\lbrace 
\begin{split}
&- \ln (1-t),  & \qquad \textrm{if $l=0$}; \\
&\cfrac{B_l}{l} + (-1)^l \sum\limits_{k=1}^l (k-1)! \left( \frac{t}{1-t} \right)^k \gamma(l,k), & \qquad \textrm{if } l \geq 1. \\
\end{split}\right. 
\end{equation}
\noindent
Here $B_l$ is the Bernoulli number and the number $\gamma(l,k)$ is defined by the generating function
\begin{equation*}
\sum_{l \geq 0} \gamma(l,k) \frac{z^l}{l!} := \frac{(e^z-1)^k}{k!}.
\end{equation*}
We notice that $\gamma(l,k)$ vanishes for $k>l$ and that the sum over $l$ in \eqref{cvirx} is finite because the $l$-th Chern character $\Ch_l$ vanishes for $l > \dim(S)$. By~\cite[Lemma 3.18]{Guere1}, the definition \eqref{cvirx} coincides with \eqref{charclass2} when $\left| t \right| < 1$. We see that the function $t \mapsto \fc_t$ is a meromorphic function with coefficients in $H^*(S)$ and with a unique pole at $t=1$.

Now we can state the main theorem of \cite{Guere2,PhDJG} in the case of the $r$-spin theory. 
\begin{thm}\label{thmprodvirt}
For any genus $g$ and any numbers $m_1,\dotsc,m_n \in \left\lbrace 1,\dotsc,r-1 \right\rbrace$, the limit $\lim_{t\to 1}\fc_t(-R^\bullet \pi_*(\cL))\fc_{t^{-r}}(\mathbb{E}^\vee)$ exists and we have
\begin{equation}\label{lW}
(-1)^g ~ \lambda_g ~ \cvir(m_1,\dotsc,m_n)_{g,n}=\lim_{t\to 1}\fc_t(-R^\bullet \pi_*(\cL))\fc_{t^{-r}}(\mathbb{E}^\vee),
\end{equation}
where $\mathbb E$ is the Hodge vector bundle over $\sS^r_{g,n}(m_1,\dotsc,m_n)$.
\end{thm}

Mumford's formula~\cite{Mumford} expresses the Chern character of the Hodge bundle over~$\oM_{g,n}$ in terms of tautological classes. Chiodo's formula \cite[Theorem 1.1.1]{Chiodo1} is a generalization of it and computes the Chern character of the higher push-forward~$R^\bullet \pi_*(\cL)$ over the moduli space~$\sS^r_{g,n}(m_1,\ldots,m_n)$. Together with Theorem \ref{thmprodvirt} and definitions \eqref{cvirx} and \eqref{parametresl}, we are then able to compute explicitly the class \eqref{virtprod} in terms of tautological classes.

\begin{rem}
An alternative way to compute the push-forward of the left-hand side of equation~\eqref{lW} to the moduli space of stable curves is to use Teleman's classification of semi-simple cohomological field theories \cite{Teleman}.
Our method has nevertheless the advantage to be easier to implement into a computer (see \cite{computerprogram,PhDJG}).
Moreover, we notice that equation \eqref{lW} is more general, since it is valid even in the Chow ring of the moduli space $\sS^r_{g,n}$, and that our approach does not use any semi-simplicity condition.
\end{rem}

In particular, the second author has developed a computer program \cite{computerprogram,PhDJG} evaluating any Hodge integrals, i.e.~intersection numbers involving $\psi$-classes, the class $\lambda_g$ and Witten's class.
All the numerical results in this paper have been obtained with it. We give an example of such computations in the next section, without providing all the details. However, it is useful to write Chiodo's formula
\begin{equation}\label{grrorb}
\mathrm{Ch}(R^\bullet\pi_*\cL) = \sum_{l \geq 0} \left( \frac{B_{l+1}(\frac{1}{r})}{(l+1)!} \kappa_l - \sum_{i=1}^n \frac{B_{l+1}(\frac{m(i)}{r})}{(l+1)!} \psi_i^l + \frac{r}{2} \sum_{m=1}^{r-1} \frac{B_{l+1}(\frac{m}{r})}{(l+1)!} (j_m)_*(\delta_{l-1}) \right).
\end{equation}
Here $B_d(x)$ is the Bernoulli polynomial defined by the generating series
\begin{equation*}
\sum_{d \geq 0} B_d(x) \frac{t^d}{d!} = \frac{t e^{xt}}{e^t-1} \qquad \textrm{(with $B_l:=B_l(0)$)}.
\end{equation*}
The morphism $j_m$ goes from $\widetilde{\Delta}_m$ to the moduli space $\sS^r_{g,n}(m_1,\dotsc,m_n)$, where the divisor $\Delta_m$ corresponds to the nodal $r$-spin curves with a node of multiplicity $m$ and $\widetilde{\Delta}_m$ is the covering of $\Delta_m$ corresponding to the extra choice of a node together with a branch with multiplicity $m$ at the node.
Note that the multiplicity at the node on the other branch is $r-m$.
The class $\delta_l$ is defined by
\begin{equation*}
\delta_l=
\begin{cases}
\sum_{a+a'=l}\psi^a(-\hat{\psi})^{a'},&\text{if $l\ge 0$},\\
0,&\text{otherwise},
\end{cases}
\end{equation*}
where $\psi$ is the first Chern class of the line bundle on $\widetilde{\Delta}_m$ corresponding to the cotangent line bundle at the given node and on the chosen branch (the one with multiplicity $m$), and $\hat{\psi}$ is the first Chern class of the line bundle on $\widetilde{\Delta}_m$ corresponding to the cotangent line bundle at the given node on the other branch (the one with multiplicity $r-m$).

\subsection{Double ramification hierarchy for the $r$-spin theory}

In order to compute the Hamiltonians~\eqref{DR Hamiltonians} of the double ramification hierarchy, we have to compute the integrals
\begin{equation}\label{DR integrals another}
\int_{\DR_g\left(0,a_1,\ldots,a_n\right)}\lambda_g\psi_1^d c^{\text{$r$-spin}}_{g,n+1}\left(e_\alpha\otimes \bigotimes_{i=1}^n e_{\alpha_i}\right).
\end{equation}
By checking the degree, this integral can be non-zero only if
\begin{eqnarray*}
3g-3+n+1 & = & 2g + d + \degvir \nonumber \\
& = & 2g + d + \frac{\sum_{i=1}^n \alpha_i + \alpha - n -1 + 2 - 2g}{r}+g-1,
\end{eqnarray*}
or, equivalently,
\begin{gather}\label{eq:selection for DR integrals}
\sum_{i=1}^n\alpha_i=(r+1)n+(2g-1-\alpha)-r(d+1).
\end{gather}
Using the inequality $\alpha_i\le r-1$, we obtain
\begin{gather}\label{eq:inequality}
2n+2g\le\alpha+1+r(d+1).
\end{gather}
Thus, for any fixed~$r,\alpha$ and~$d$ we have a finite number of choices for~$g,n$ and, thus, for~$\alpha_i$'s.

Consider now the integral~\eqref{DR integrals another} with fixed~$r,\alpha,d$ and~$\alpha_1,\ldots,\alpha_n$ that satisfy equation~\eqref{eq:selection for DR integrals}. Hain's formula~\cite{Hai11} together with the result of~\cite{MW13} implies that
\begin{gather}\label{eq:Hain's formula}
\left.\DR_g(b_1,\dotsc,b_n)\right|_{\cM^{\ct}_{g,n}}=\frac{1}{g!} \left( \sum_{j=1}^n \frac{b_j^2 \psi^{\dagger}_j}{2} - \sum_{\substack{J \subset \left\lbrace 1,\dotsc,n\right\rbrace  \\ \left| J \right| \geq 2}} \left(\sum_{i,j \in J, i<j} b_i b_j\right) \delta_0^J - \frac{1}{4} \sum_{J \subset \left\lbrace 1,\dotsc,n\right\rbrace} \sum_{h=1}^{g-1} b_J^2 \delta_h^J \right)^g,
\end{gather}
where~$\psi^\dagger_j$ denotes the $\psi$-class that is pulled back from~$\oM_{g,1}$, the integer $b_J$ is the sum $\sum_{j\in J} b_j$ and the class $\delta_h^J$ represents the divisor whose generic point is a nodal curve made of one smooth component of genus $h$ with the marked points labeled by the list $J$ and of another smooth component of genus $g-h$ with the remaining marked points, joined at a separating node. It is often more convenient to rewrite formula~\eqref{eq:Hain's formula} using the usual $\psi$-classes:
\begin{gather}\label{eq:Hain's formula,second}
\left.\DR_g(b_1,\dotsc,b_n)\right|_{\cM^{\ct}_{g,n}}=\frac{1}{g!}\left( \sum_{j=1}^n \frac{b_j^2 \psi_j}{2}-\frac{1}{2}\sum_{\substack{J \subset \left\lbrace 1,\dotsc,n\right\rbrace  \\ \left|J\right|\geq 2}}b_J^2 \delta_0^J - \frac{1}{4} \sum_{J \subset \left\lbrace 1,\dotsc,n\right\rbrace} \sum_{h=1}^{g-1} b_J^2 \delta_h^J \right)^g,
\end{gather}
Since the class $\lambda_g$ vanishes on the complement of $\mathcal{M}_{g,n}^\textrm{ct}$ in $\overline{\mathcal{M}}_{g,n}$, we have
\begin{multline*}
\int_{\DR_g(0,a_1,\dotsc,a_n)}\lambda_g\psi_1^d c^{\text{$r$-spin}}_{g,n+1}\left(e_\alpha\otimes \bigotimes_{i=1}^n e_{\alpha_i}\right)=\\
=\frac{1}{g!}\int_{\oM_{g,n+1}}\lambda_g\psi_1^d c^{\text{$r$-spin}}_{g,n+1}\left(e_\alpha\otimes \bigotimes_{i=1}^n e_{\alpha_i}\right)\left(\sum_{j=1}^n \frac{a_j^2 \psi_j}{2}-\frac{1}{2}\sum_{\substack{J \subset \left\lbrace 0,1,\dotsc,n\right\rbrace  \\ \left| J \right| \geq 2}}a_J^2\delta_0^J-\frac{1}{4}\sum_{J \subset \left\lbrace 0,1,\dotsc,n\right\rbrace} \sum_{h=1}^{g-1} a_J^2 \delta_h^J \right)^g,
\end{multline*}
where on the right-hand side we, by definition, put $a_0:=0$. Applying the procedure, described in the previous section, we are then able to compute this integral. As a result, we obtain an algorithm for the computation of an arbitrary fixed Hamiltonian of the double ramification hierarchy for the $r$-spin theory.

\subsection{Computation of the Hamiltonian $\og_{1,1}$}\label{cg11}

For the Hamiltonian~$\og_{1,1}$, the inequality~\eqref{eq:inequality} becomes
\begin{equation*}
g+n\leq r+1.
\end{equation*}
Since we require $n \geq 2$ in the definition of the Hamiltonian~\eqref{DR Hamiltonians}, we must have
\begin{equation*}
g \leq r-1.
\end{equation*}
Note that the case $g=r-1$ happens exactly for $n=2$ and $\alpha_1=\alpha_2=r-1$. 

Denote by $n_k$ the cardinal of the set $\left\lbrace 1 \leq i \leq n \left| \right. \alpha_i=k \right\rbrace$, so that equation \eqref{eq:selection for DR integrals} becomes
\begin{equation}\label{selec3}
2g-2 + \sum_{k=1}^{r-1} (r+1-k) n_k = 2r.
\end{equation}
In particular, for $r \leq 5$, the only solutions to \eqref{selec3} with $n=\sum_{k=1}^{r-1}n_k\geq 2$ are
\begin{align*}
&r=3 : \quad \left\lbrace \begin{array}{lcl}
g=0 & \textrm{and} & (n_1,n_2) \in \left\lbrace (0,4), (2,1) \right\rbrace, \\
g=1 & \textrm{and} & (n_1,n_2) \in \left\lbrace (0,3), (2,0) \right\rbrace, \\
g=2 & \textrm{and} & (n_1,n_2) \in \left\lbrace (0,2) \right\rbrace; \\
\end{array}\right.\\
&r=4 : \quad \left\lbrace \begin{array}{lcl}
g=0 & \textrm{and} & (n_1,n_2,n_3) \in \left\lbrace (0, 0, 5), (0, 2, 2), (1, 0, 3), (1, 2, 0), (2, 0, 1) \right\rbrace, \\
g=1 & \textrm{and} & (n_1,n_2,n_3) \in \left\lbrace (0, 0, 4), (0, 2, 1), (1, 0, 2), (2, 0, 0) \right\rbrace, \\
g=2 & \textrm{and} & (n_1,n_2,n_3) \in \left\lbrace (0, 0, 3), (0, 2, 0), (1, 0, 1) \right\rbrace, \\
g=3 & \textrm{and} & (n_1,n_2,n_3) \in \left\lbrace (0, 0, 2) \right\rbrace; \\
\end{array}\right.\\
&r=5 : \quad \left\lbrace \begin{array}{lcll}
g=0 & \textrm{and} & (n_1,n_2,n_3,n_4) \in &\left\lbrace (0, 0, 0, 6), (0, 0, 2, 3), (0, 0, 4, 0), (0, 1, 0, 4), (0, 1, 2, 1),\right. \\
& & & \left.  (0, 2, 0, 2), (0, 3, 0, 0), (1, 0, 1, 2), (1, 1, 1, 0), (2, 0, 0, 1) \right\rbrace, \\
g=1 & \textrm{and} & (n_1,n_2,n_3,n_4) \in &\left\lbrace (0, 0, 0, 5), (0, 0, 2, 2), (0, 1, 0, 3), (0, 1, 2, 0), (0, 2, 0, 1), \right. \\
& & & \left. (1, 0, 1, 1), (2, 0, 0, 0) \right\rbrace, \\
g=2 & \textrm{and} & (n_1,n_2,n_3,n_4) \in &\left\lbrace (0, 0, 0, 4), (0, 0, 2, 1), (0, 1, 0, 2), (0, 2, 0, 0), (1, 0, 1, 0) \right\rbrace, \\
g=3 & \textrm{and} & (n_1,n_2,n_3,n_4) \in &\left\lbrace (0, 0, 0, 3), (0, 0, 2, 0), (0, 1, 0, 1) \right\rbrace, \\
g=4 & \textrm{and} & (n_1,n_2,n_3,n_4) \in &\left\lbrace (0, 0, 0, 2) \right\rbrace. \\
\end{array}\right.
\end{align*}

Once we have found all the non-trivially-zero contributions to the Hamiltonian~$\og_{1,1}$, we have to compute the integrals
$$
\int_{\DR_g\left(0,a_1,\ldots,a_n\right)}\lambda_g\psi_1 c^{\text{$r$-spin}}_{g,n+1}\left(e_1\otimes \bigotimes_{k=1}^{r-1} e_k^{\otimes n_k}\right).
$$
Note that comparing to the general case~\eqref{DR integrals another} this integral can be simplified a little bit using the dilaton equation. We get
\begin{align*}
&\int_{\DR_g\left(0,a_1,\ldots,a_n\right)}\lambda_g\psi_1 c^{\text{$r$-spin}}_{g,n+1}\left(e_1\otimes \bigotimes_{k=1}^{r-1} e_k^{\otimes n_k}\right)=(2g-2+n)\int_{\DR_g(a_1,\dotsc,a_n)}\lambda_g c^{\text{$r$-spin}}_{g,n}\left(\bigotimes_{k=1}^{r-1} e_k^{\otimes n_k}\right)=\\
=&\frac{2g-2+n}{g!}\int_{\oM_{g,n}}\lambda_g c^{\text{$r$-spin}}_{g,n}\left(\bigotimes_{k=1}^{r-1} e_k^{\otimes n_k}\right)\left(\sum_{j=1}^n \frac{a_j^2 \psi_j}{2} - \frac{1}{2}\sum_{\substack{J \subset \left\lbrace 1,\dotsc,n\right\rbrace  \\ \left| J \right| \geq 2}}a_J^2 \delta_0^J - \frac{1}{4} \sum_{J \subset \left\lbrace 1,\dotsc,n\right\rbrace} \sum_{h=1}^{g-1} a_J^2 \delta_h^J \right)^g.
\end{align*}

As an example, we consider the contribution to the Hamiltonian~$\og_{1,1}$ corresponding to~$(g,n_1,n_2)=(2,0,2)$ for the $3$-spin theory:
\begin{equation*}
T^{\text{$3$-spin}}_{(2,0,2)}:=2\eps^4\sum_{a\in\mbZ} \left( \int_{\DR_2(a,-a)}\lambda_2 c^{\text{$3$-spin}}_{2,2}\left(e_2^{\otimes 2}\right)\right) p^2_{a}p^2_{-a}.
\end{equation*}
We get
$$
T^{\text{$3$-spin}}_{(2,0,2)}=\eps^4\sum_{a\in\mbZ}a^4 p^2_a p^2_{-a}\int_{\oM_{2,2}}\lambda_2 c^{\text{$3$-spin}}_{2,2}\left(e_2^{\otimes 2}\right)\left(\frac{\psi_1+\psi_2}{2}-\frac{1}{2}\delta_1^{\{1\}}\right)^2.
$$
Using the factorization property and the selection rule~\eqref{selecrule} we can easily see that
$$
c^{\text{$3$-spin}}_{2,2}\left(e_2^{\otimes 2}\right)\cdot\delta_1^{\{1\}}=0.
$$
Therefore, we get
\begin{multline*}
T^{\text{$3$-spin}}_{(2,0,2)}=\frac{\eps^4}{4}\sum_{a\in\mbZ}a^4 p^2_a p^2_{-a}\int_{\oM_{2,2}}\lambda_2 c^{\text{$3$-spin}}_{2,2}\left(e_2^{\otimes 2}\right)\left(\psi_1+\psi_2\right)^2=\\
=\frac{\eps^4}{2}\sum_{a\in\mbZ}a^4 p^2_a p^2_{-a}\int_{\oM_{2,2}}\lambda_2 c^{\text{$3$-spin}}_{2,2}\left(e_2^{\otimes 2}\right)\left(\psi_1^2+\psi_1\psi_2\right).
\end{multline*}
Now, we use Theorem \ref{thmprodvirt} and, with a lot of help from the computer program \cite{computerprogram,PhDJG}, we get the following values for the Hodge integrals:
\begin{gather*}
\int_{\overline{\mathcal{M}}_{2,2}}\psi_2^2\lambda_2 c^{\text{$3$-spin}}_{2,2}\left(e_2^{\otimes 2}\right)=\frac{7}{4320},\qquad \int_{\overline{\mathcal{M}}_{2,2}}\psi_2\psi_3\lambda_2 c^{\text{$3$-spin}}_{2,2}\left(e_2^{\otimes 2}\right)=\frac{13}{4320}.
\end{gather*}
As a consequence, we obtain
\begin{gather*}
T^{\text{$3$-spin}}_{(2,0,2)}=\eps^4\sum_{a\in\mbZ}\frac{a^4}{432}p^2_ap^2_{-a}=\int\frac{\eps^4}{432}u^2 u^2_4dx.
\end{gather*}

\begin{proposition}\label{proposition:DR 3-spin}
The Hamiltonian $\og^{\text{$3$-spin}}_{1,1}$ for the $3$-spin theory equals
\begin{gather*}
\og^{\text{$3$-spin}}_{1,1}=\int\left(\frac{(u^1)^2 u^2}{2}+\frac{(u^2)^4}{36}+\eps^2\left(\frac{(u^2)^2 u^2_2}{48}+\frac{u^1 u^1_2}{12}\right)+\frac{\eps^4}{432}u^2 u^2_4\right)dx.
\end{gather*}
\end{proposition}

\begin{proof}
It is a direct computation using the computer program \cite{computerprogram,PhDJG}.
\end{proof}

With the same method, we compute the Hamiltonians~$\og^{\text{$4$-spin}}_{1,1}$ and~$\og^{\text{$5$-spin}}_{1,1}$ for the~$4$ and $5$-spin theories.

\begin{proposition}\label{proposition:DR 4-spin}
The Hamiltonian $\og^{\text{$4$-spin}}_{1,1}$ for the $4$-spin theory equals
\begin{align*}
\og^{\text{$4$-spin}}_{1,1}=&\int\left[\frac{(u^1)^2u^3}{2}+\frac{u^1(u^2)^2}{2}+\frac{(u^2)^2(u^3)^2}{8}+\frac{(u^3)^5}{320}+\right.\\
&\phantom{\int a}\eps^2\left(\frac{1}{8}u^1u^1_2+\frac{1}{64}u^3_2(u^2)^2+\frac{1}{16}u^3u^2u^2_2+\frac{1}{64}u^1_2(u^3)^2+\frac{1}{192}(u^3)^3u^3_2\right)+\\
&\phantom{\int a}\eps^4\left(\frac{1}{160}u^2u^2_4+\frac{5}{4096}(u^3)^2u^3_4+\frac{3}{640}u^1u^3_4\right)+\\
&\phantom{\int a}\eps^6\left.\frac{1}{8192}u^3u^3_6\right]dx.
\end{align*}
\end{proposition}

\begin{rem}
Above Hamiltonians~$\og^{\text{$3$-spin}}_{1,1}$ and $\og^{\text{$4$-spin}}_{1,1}$ are equal to the computations done in~\cite[Examples 4.3 and 4.4]{BR14}. The approach there is different and uses recursion properties of the double ramification hierarchy. Our computation can be seen as a second check. However, the Hamiltonian~$\og^{\text{$5$-spin}}_{1,1}$ of the following proposition is new.
\end{rem}

\begin{proposition}\label{proposition:DR 5-spin}
The Hamiltonian $\og^{\text{$5$-spin}}_{1,1}$ for the $5$-spin theory equals
\begin{align*}
\og^{\text{$5$-spin}}_{1,1}=&\int\left[\frac{(u^1)^2 u^4}{2}+u^1 u^2 u^3+\frac{(u^2)^3}{6}+\frac{(u^3)^4}{30}+\frac{u^2(u^3)^2 u^4}{5}+\frac{(u^2)^2 (u^4)^2}{10}+\frac{(u^3)^2 (u^4)^3}{50}+\frac{(u^4)^6}{3750}+\right.\\
&\phantom{\int a}\eps^2\left(\frac{1}{6}u^1u^1_2+\frac{3}{20}u^2 u^3 u^3_2+\frac{1}{10}u^2(u^3_1)^2+\frac{1}{20}u^1_2 u^3 u^4+\frac{1}{10}u^2 u^2_2 u^4+\frac{1}{40} (u^2_1)^2 u^4\right. \\
&\left.\phantom{\int a \hbar}+\cfrac{1}{50} u^2 u^4(u^4_1)^2+\frac{1}{75} u^2 (u^4)^2 u^4_2+\frac{1}{75}(u^3)^2 u^4u^4_2+\frac{1}{50} u^3 u^3_2(u^4)^2+\frac{1}{1200}(u^4)^4 u^4_2\right)+\\
&\phantom{\int a}\eps^4\left(\frac{7}{600}u^2 u^2_4+\frac{11}{900}u^1 u^3_4+\frac{7}{1200}u^2 u^4 u^4_4+\frac{17}{1200}u^2 u^4_1 u^4_3+\frac{71}{7200}u^2 (u^4_2)^2+\frac{31}{3600}u^3 u^3_4 u^4\right.\\
&\left.\phantom{\int a \hbar^2}+\frac{7}{450}u^3_1 u^3_3 u^4+\frac{91}{7200}(u^3_2)^2 u^4+\frac{13}{12000}(u^4_2)^2(u^4)^2+\frac{3}{4000}u^4_2 (u^4_1)^2 u^4\right)+\\
&\phantom{\int a}\eps^6\left(\frac{53}{108000}u^3 u^3_6+\frac{11}{18000}u^2 u^4_6+\frac{1397}{6480000}(u^4_3)^2 u^4+\frac{617}{1620000}u^4_4 u^4_2 u^4\right)+\\
&\phantom{\int a}\eps^8\left.\frac{107}{10800000}u^4 u^4_8\right]dx.
\end{align*}
\end{proposition}


\section{Dubrovin--Zhang hierarchy for the $r$-spin theory}\label{section:DZ hierarchy for r-spin}

In this section we review the description of the Dubrovin--Zhang hierarchy for the $r$-spin theory (\cite{Witten2,FSZ,DZ05}). In Section~\ref{subsection:Miura transformations} we discuss Miura transformations of hamiltonian hierarchies and fix some notations. In Section~\ref{subsection:pseudo-differential operators} we recall basic facts about pseudo-differential operators. In Section~\ref{subsection:GD hierarchy} we review the construction of the $r$-th Gelfand--Dickey hierarchy. In Section~\ref{subsection:DZ for r-spin} we describe the Dubrovin--Zhang hierarchy for the $r$-spin theory and do some explicit computations for $r=2,3,4,5$.

\subsection{Miura transformations in the theory of hamiltonian hierarchies}\label{subsection:Miura transformations}

Here we want to discuss changes of variables in the theory of hamiltonian hierarchies and introduce appropriate notations. We recommend the reader the paper~\cite{DZ05} for a more detailed introduction to this subject. 

First of all, let us modify our notations a little bit. Recall that by $\cA_N$ we denoted the ring of differential polynomials in the variables $u^1,\ldots,u^N$. Since we are going to consider rings of differential polynomials in different variables, we want to see the variables in the notation. So for the rest of the paper we denote by~$\cA_{u^1,\ldots,u^N}$ the ring of differential polynomials in variables~$u^1,\ldots,u^N$. The same notation is adopted for the extension $\hcA_{u^1,\ldots,u^N}$ and for the spaces of local functionals~$\Lambda_{u^1,\ldots,u^N}$ and~$\hLambda_{u^1,\ldots,u^N}$.

Consider changes of variables of the form
\begin{align}
&\tu^\alpha(u;u_x,u_{xx},\ldots;\eps)=u^\alpha+\sum_{k\ge 1}\eps^k f^\alpha_k(u;u_x,\ldots,u_k),\quad \alpha=1,\ldots,N,\label{eq:Miura transformation}\\
&f^\alpha_k\in\cA_{u^1,\ldots,u^N},\quad\deg f^\alpha_k=k.\label{eq:degree condition}
\end{align}
They are called Miura transformations. It is not hard to see that they are invertible.

Any differential polynomial $f(u)\in\hcA_{u^1,\ldots,u^N}$ can be rewritten as a differential polynomial in the new variables $\tu^\alpha$. The resulting differential polynomial is denoted by $f(\tu)$. The last equation in line~\eqref{eq:degree condition} garanties that, if $f(u)\in\hcA_{u^1,\ldots,u^N}^{[d]}$, then $f(\tu)\in\hcA_{\tu^1,\ldots,\tu^N}^{[d]}$. In other words, a Miura transformation defines an isomorphism $\hcA_{u^1,\ldots,u^N}^{[d]}\simeq\hcA_{\tu^1,\ldots,\tu^N}^{[d]}$. In the same way any Miura transformation identifies the spaces of local functionals $\hLambda^{[d]}_{u^1,\ldots,u^N}$ and $\hLambda^{[d]}_{\tu^1,\ldots,\tu^N}$. For any local functional $\oh[u]\in\hLambda^{[d]}_{u^1,\ldots,u^N}$ the image of it under the isomorphism $\hLambda^{[d]}_{u^1,\ldots,u^N}\stackrel{\sim}{\to}\hLambda^{[d]}_{\tu^1,\ldots,\tu^N}$ is denoted by $\oh[\tu]\in\hLambda^{[d]}_{\tu^1,\ldots,\tu^N}$. 

Let us describe the action of Miura transformations on hamiltonian hierarchies. Suppose we have a hamiltonian system
\begin{gather}\label{eq:Hamiltonian system2}
\frac{\d u^\alpha}{\d\tau_i}=K^{\alpha\mu}\frac{\delta\oh_i[u]}{\delta u^\mu},\quad\alpha=1,\ldots,N,\quad i\ge 1,
\end{gather}
defined by a hamiltonian operator $K$ and a sequence of pairwise commuting local functionals $\oh_i[u]\in\hLambda^{[0]}_{u^1,\ldots,u^N}$, $\{\oh_i[u],\oh_j[u]\}_K=0$. Consider a Miura transformation~\eqref{eq:Miura transformation}. Then in the new variables~$\tu^\alpha$, the system~\eqref{eq:Hamiltonian system2} looks as follows:
\begin{align}
&\frac{\d\tu^\alpha}{\d\tau_i}=K_{\tu}^{\alpha\mu}\frac{\delta\oh_i[\tu]}{\delta \tu^\mu},\quad\text{where}\notag\\
&K_{\tu}^{\alpha\beta}=\sum_{p,q\ge 0}\frac{\d \tu^\alpha(u)}{\d u^\mu_p}\d_x^p\circ K^{\mu\nu}\circ(-\d_x)^q\circ\frac{\d \tu^\beta(u)}{\d u^\nu_q}.\label{eq:transformation of an operator}
\end{align}

\subsection{Pseudo-differential operators}\label{subsection:pseudo-differential operators}

The material of this and the next sections is borrowed from the book~\cite{Dic03}. 

Let us fix $r\ge 2$ and consider variables~$f_0,f_1,\ldots,f_{r-2}$. A pseudo-differential operator $A$ is a Laurent series
$$
A=\sum_{n=-\infty}^m a_n\d_x^n,
$$
where $m$ is an arbitrary integer and $a_n\in\cA_{f_0,f_1,\ldots,f_{r-2}}$ are differential polynomials. Let
\begin{gather*}
A_+:=\sum_{n=0}^m a_n\d_x^n,\qquad \res A:=a_{-1}.
\end{gather*}
The product of pseudo-differential operators is defined by the following commutation rule:
\begin{gather*}
\d_x^k\circ a:=\sum_{l=0}^\infty\frac{k(k-1)\ldots(k-l+1)}{l!}(\d_x^l a)\d_x^{k-l},
\end{gather*}
where $k\in\Z$ and $a\in\cA_{f_0,f_1,\ldots,f_{r-2}}$. For any $m\ge 2$ and a pseudo-differential operator~$A$ of the form
$$
A=\d_x^m+\sum_{n=1}^\infty a_n\d_x^{m-n},
$$
there exists a unique pseudo-differential operator $A^{\frac{1}{m}}$ of the form
$$
A^{\frac{1}{m}}=\d_x+\sum_{n=0}^\infty \widetilde{a}_n\d_x^{-n},
$$
such that $\left(A^{\frac{1}{m}}\right)^m=A$.

\subsection{Gelfand--Dickey hierarchy}\label{subsection:GD hierarchy}

Let 
$$
L:=\d_x^r+f_{r-2}\d_x^{r-2}+\ldots+f_1\d_x+f_0.
$$
The $r$-th Gelfand--Dickey hierarchy is the following system of partial differential equations: 
\begin{gather}\label{eq:GD hierarchy}
\frac{\d L}{\d T_m}=[(L^{m/r})_+,L],\quad m\ge 1.
\end{gather}
We immediately see that $\frac{\d L}{\d T_{rk}}=0$, so we can omit the times~$T_{rk}$. Since $(L^{1/r})_+=\d_x$, we have $\frac{\d f_i}{\d T_1}=(f_i)_x$. 

The Gelfand--Dickey hierarchy has two compatible hamiltonian structures. The second one is not needed in this paper, so we recall only the first one. Let $X_0,X_1,\ldots,X_{r-2}\in\cA_{f_0,\ldots,f_{r-2}}$ be some differential polynomials. Consider a pseudo-differential operator 
$$
X:=\d_x^{-(r-1)}\circ X_{r-2}+\ldots+\d_x^{-1}\circ X_0.
$$ 
It is easy to see that the positive part $[X,L]_+$ of the commutator has the following form:
$$
[X,L]_+=\sum_{0\le\alpha,\beta\le r-2}((K^\GD)^{\alpha\beta}X_\beta)\d_x^\alpha,
$$
where
$$
(K^{\GD})^{\alpha\beta}=\sum_{i\ge 0}(K^{\GD})^{\alpha\beta}_i\d_x^i,\quad (K^{\GD})^{\alpha\beta}_i\in\cA_{f_0,\ldots,f_{r-2}},
$$
are differential operators and the sum is finite. The operator $K^{\GD}=((K^{\GD})^{\alpha\beta})_{0\le\alpha,\beta\le r-2}$ is hamiltonian. Consider local functionals
$$
\oh_m^{\GD}:=-\frac{r}{m+r}\int \res L^{(m+r)/r}dx,\quad m\ge 1.
$$
We have 
$$
\left\{\oh^{\GD}_m,\oh^{\GD}_n\right\}_{K^{\GD}}=0.
$$
For a local functional $\oh\in\Lambda_{f_0,f_1,\ldots,f_{r-2}}$ define a pseudo-differential operator $\frac{\delta\oh}{\delta L}$ by
$$
\frac{\delta\oh}{\delta L}:=\d_x^{-(r-1)}\circ\frac{\delta\oh}{\delta f_{r-2}}+\ldots+\d_x^{-1}\circ\frac{\delta\oh}{\delta f_0}.
$$
Then the right-hand side of~\eqref{eq:GD hierarchy} can be written in the following way: 
$$
[(L^{m/r})_+,L]=\left[\frac{\delta\oh^{\GD}_m}{\delta L},L\right]_+=\sum_{0\le\alpha,\beta\le r-2}\left((K^\GD)^{\alpha\beta}\frac{\delta\oh^{\GD}_m}{\delta f_\beta}\right)\d_x^\alpha.
$$
Therefore, the sequence of local functionals $\oh^{\GD}_m$ together with the hamiltonian operator~$K^{\GD}$ define a hamiltonian structure of the Gelfand--Dickey hierarchy~\eqref{eq:GD hierarchy}.

\subsection{Dubrovin--Zhang hierarchy for the $r$-spin theory}\label{subsection:DZ for r-spin}

Introduce new variables $w^1,\ldots,w^{r-1}$ by
$$
w^\alpha=\frac{1}{(r-\alpha)(-r)^{\frac{r-\alpha-1}{2}}}\res L^{(r-\alpha)/r}.
$$
Define a hamiltonian operator $K^{\text{$r$-spin}}=((K^{\text{$r$-spin}})^{\alpha\beta})_{1\le\alpha,\beta\le r-1}$ and local functionals $\oh^{\text{$r$-spin}}_{\alpha,d}\in\Lambda_{w^1,\ldots,w^{r-1}}$, $1\le\alpha\le r-1$, $d\ge 0$, by
\begin{align*}
&K^{\text{$r$-spin}}:=(-r)^{\frac{r}{2}}K^{\GD}_w,\\
&\oh_{\alpha,d}^{\text{$r$-spin}}:=\frac{1}{(-r)^{\frac{r+k-1}{2}-d}k!_r}\oh_k^{\GD}[w],
\end{align*}
where $k:=\alpha+rd$ and $k!_r:=\prod_{i=0}^d(\alpha+ri)$. Recall that $K^\GD_w$ denotes the Miura transform of the operator $K^{\GD}$ that is described by formula~\eqref{eq:transformation of an operator}. Then the Dubrovin--Zhang hierarchy for the $r$-spin theory is given by the sequence of local functionals~$\oh^{\text{$r$-spin}}_{\alpha,d}$ and the hamiltonian operator~$K^{\text{$r$-spin}}$. 

\subsection{Examples}\label{subsection:examples}

Here we compute the Hamiltonian~$\oh^{\text{$r$-spin}}_{1,1}$ and the operator~$K^{\text{$r$-spin}}$ for $r=2,3,4,5$. When we present the final answer in these cases, just for convenience, we recover the parameter~$\eps$.

\subsubsection{$2$-spin theory}

Denote $f_0$ by $f$ and $w^1$ by $w$. We compute
\begin{align*}
L=&\d_x^2+f,\\
\res L^{5/2}=&\frac{5}{16}f^3+\frac{5}{32}f_x^2+\frac{5}{16}ff_{xx}+\frac{1}{32}f_{xxxx},\\
\oh_3^{\GD}=&\int\left(-\frac{1}{8}f^3-\frac{1}{16}ff_{xx}\right)dx,\\
K^{\GD}=&-2\d_x.
\end{align*}
The variable $w$ is related to the variable $f$ by $w=\frac{f}{2}$. As a result, for the Dubrovin--Zhang hierarchy we get
\begin{align*}
K^{\text{$2$-spin}}=&\d_x,\\
\oh^{\text{$2$-spin}}_{1,1}=&\int\left(\frac{w^3}{6}+\eps^2\frac{w w_{xx}}{24}\right)dx.
\end{align*}

\subsubsection{$3$-spin theory}

We have
\begin{align*}
L=&\d_x^3+f_1\d_x+f_0,\\
K^{\GD}=&\begin{pmatrix}
0      & -3\d_x\\
-3\d_x & 0
\end{pmatrix},\\
\oh^{\GD}_4=&\int\left(-\frac{2}{9}f_0^2f_1+\frac{1}{81}f_1^4-\frac{1}{9}f_0(f_0)_{xx}+\frac{2}{9}f_0f_1(f_1)_x+\frac{1}{18}f_1^2(f_1)_{xx}+\frac{1}{9}f_0(f_1)_{xxx}\right.\\
&\phantom{\int (}\left.+\frac{1}{27}f_1(f_1)_{xxxx}\right)dx.
\end{align*}
The relation between the variables $w^1,w^2$ and $f_0,f_1$ looks as follows:
\begin{gather*}
\left\{
\begin{aligned}
&w^1=\frac{1}{2\sqrt{-3}}\left(\frac{2}{3}f_0-\frac{1}{3}(f_1)_x\right),\\
&w^2=\frac{f_1}{3}.
\end{aligned}
\right.
\end{gather*}
For the Dubrovin--Zhang hierarchy we obtain
\begin{align*}
K^{\text{$3$-spin}}=&
\begin{pmatrix}
0           & \d_x \\
\d_x & 0 
\end{pmatrix},\\
\oh^{\text{$3$-spin}}_{1,1}=&\int\left[\frac{(w^2)^4}{36}+\frac{w^2(w^1)^2}{2}+\eps^2\left(\frac{(w^2)^2w^2_{xx}}{48}+\frac{w^1w^1_{xx}}{12}\right)+\eps^4\frac{w^2w^2_{xxxx}}{432}\right]dx.
\end{align*}

\subsubsection{$4$-spin theory}

For the Dubrovin--Zhang hierarchy we get
\begin{align*}
K^{\text{$4$-spin}}=&\begin{pmatrix}
\frac{1}{48}\eps^2\d_x^3 & 0    & \d_x \\
0                        & \d_x & 0    \\
\d_x                     & 0    & 0
\end{pmatrix},\\
\oh_{1,1}^{\text{$4$-spin}}=&\int\left[{\frac{w^1 (w^2)^2}{2}+\frac{(w^1)^2 w^3}{2}+\frac{(w^2)^2 (w^3)^2}{8}+\frac{(w^3)^5}{320}}+\right. \\
&\phantom{\int a}\eps^2\left(\frac{w^1 w^1_2}{8}+\frac{w^1 w^3 w^3_2}{48}+\frac{w^1(w^3_1)^2}{32}+\frac{w^2 w^3 w^2_2}{12}+\frac{w^3 (w^2_1)^2}{48}+\frac{(w^3)^3 w^3_2}{64}+\frac{(w^3)^2(w^3_1)^2}{32}\right)+\\
&\phantom{\int a}\eps^4\left(\frac{w^2 w^2_4}{160}+\frac{w^1w^3_4}{480}+\frac{5}{4608}(w^3)^2w^3_4\right)+\\
&\phantom{\int a}\eps^6\left.\frac{w^3w^3_6}{11520}\right]dx.
\end{align*}

\subsubsection{$5$-spin theory}

For the Dubrovin--Zhang hierarchy we have
\begin{align*}
K^{\text{$5$-spin}}=&\begin{pmatrix}
0                        & \frac{1}{30}\eps^2\d_x^3 & 0    & \d_x \\
\frac{1}{30}\eps^2\d_x^3 & 0                        & \d_x & 0    \\
0                        & \d_x                     & 0    & 0    \\
\d_x                     & 0                        & 0    & 0   
\end{pmatrix},\\
\oh_{1,1}^{\text{$5$-spin}}=&\int\left[\frac{\left(w^1\right)^2 w^4}{2}+w^1 w^2 w^3+\frac{\left(w^2\right)^3}{6}+\frac{\left(w^2\right)^2 \left(w^4\right)^2}{10}+\frac{w^2\left(w^3\right)^2 w^4}{5}+\frac{\left(w^3\right)^4}{30}+\frac{\left(w^3\right)^2 \left(w^4\right)^3}{50}\right.\\
&\phantom{\int a}+\frac{\left(w^4\right)^6}{3750}+\\
&\phantom{\int a}\eps^2\left(\frac{w^4_2\left(w^4\right)^4}{1200}+\frac{w^4_2 w^2 \left(w^4\right)^2}{100}+\frac{w^3_2 w^3 \left(w^4\right)^2}{50}+\frac{\left(w^2_1\right)^2 w^4}{120}+\frac{w^4_2 \left(w^3\right)^2 w^4}{100}+\frac{\left(w^4_1\right)^2 w^2 w^4}{50}\right.\\
&\phantom{\int a \eps^2}+\left.\frac{w^2_2 w^2 w^4}{12}+\frac{w^1_2 w^3 w^4}{30}+\frac{w^1_2 w^1}{6}+\frac{w^3_1 w^4_1 w^1}{30}+\frac{\left(w^3_1\right)^2 w^2}{10}+\frac{2}{15} w^3_2 w^2 w^3\right)+\\
&\phantom{\int a}\eps^4\left(\frac{w^4_4 \left(w^4\right)^3}{14400}+\frac{49\left(w^4_2\right)^2 \left(w^4\right)^2}{72000}+\frac{13\left(w^3_2\right)^2 w^4}{1800}+\frac{7}{900}w^3_1 w^3_3 w^4+\frac{1}{300} w^4_4 w^2 w^4\right.\\
&\phantom{\int a \eps^4}+\left.\frac{1}{180} w^3_4 w^3 w^4+\frac{1}{150} w^3_4 w^1+\frac{1}{120}\left(w^4_2\right)^2 w^2+\frac{7}{600}w^2_4 w^2+\frac{7}{600}w^4_1 w^4_3 w^2\right)+\\
&\phantom{\int a}\eps^6\left(\frac{178 w^4 \left(w^4_3\right)^2}{10125}-\frac{589 w^4_6 \left(w^4\right)^2}{135000}+\frac{w^4_6 w^2}{4500}+\frac{w^3_6 w^3}{3000}+\frac{1069 w^4_2 w^4_4 w^4}{40500}\right)+\\
&\phantom{\int a}\eps^8\left.\left(\frac{w^4_8 w^4}{337500}\right)\right]dx.
\end{align*}


\section{Proof of Theorem~\ref{theorem:main}}\label{section:proof}

Before proving Theorem~\ref{theorem:main} we present two simple general results that, we believe, have an independent interest. In Section~\ref{subsection:dilaton equation} we prove that the string solution of an arbitrary double ramification hierarchy satisfies the dilaton equation. In Section~\ref{subsection:reconstruction} we prove that under some assumptions a hamiltonian hierarchy can be reconstructed from its dispersionless part and only one Hamiltonian. Finally, in Section~\ref{subsection:final proof} we prove Theorem~\ref{theorem:main}.

\subsection{Dilaton equation for the string solution}\label{subsection:dilaton equation}

Consider an arbitrary cohomological field theory, $c_{g,n}\colon V^{\otimes n}\to H^{\even}(\oM_{g,n};\mbC)$, and the associated double ramification hierarchy. As usual, we denote by~$\og_{\alpha,d}$ its Hamiltonians. Let $(u^{\str})^\alpha(x,t^*_*;\eps)$ be the string solution of the double ramification hierarchy ~(see~\cite{Bur14}). Recall that it is defined as a unique solution that satisfies the initial condition~$\left.(u^{\str})^\alpha\right|_{t^*_*=0}=\delta^{\alpha,1}x$.

\begin{proposition}\label{proposition:dilaton}
We have 
\begin{gather}\label{eq:dilaton}
\frac{\d(u^{\str})^\alpha}{\d t^1_1}-\eps\frac{\d(u^{\str})^\alpha}{\d\eps}-x\frac{\d(u^{\str})^\alpha}{\d x}-\sum_{n\ge 0}t^\gamma_n\frac{\d (u^{\str})^\alpha}{\d t^\gamma_n}=0.
\end{gather}
\end{proposition}
\begin{proof}
Let $O:=\frac{\d}{\d t^1_1}-\eps\frac{\d}{\d\eps}-x\frac{\d}{\d x}-\sum_{n\ge 0}t^\gamma_n\frac{\d}{\d t^\gamma_n}$. First of all, let us check that 
\begin{gather}\label{eq:initial1}
\left.(O(u^{\str})^{\alpha})\right|_{t^*_*=0}=0.
\end{gather}
We have
\begin{gather}\label{eq:tmp1}
\left.(O(u^{\str})^{\alpha})\right|_{t^*_*=0}=\left.\left(\frac{\d(u^{\str})^\alpha}{\d t^1_1}-x\frac{\d(u^{\str})^\alpha}{\d x}\right)\right|_{t^*_*=0}=\left.\left(\frac{\d(u^{\str})^\alpha}{\d t^1_1}\right)\right|_{t^*_*=0}-\delta^{\alpha,1}x.
\end{gather}
In order to proceed, we do the same trick, as in the proof of Lemma~5.1 in~\cite{BR14}. We consider new variables $v^\alpha_d$, $1\le\alpha\le N, d\ge 0$, such that $u^\alpha_d=v^\alpha_{d+1}$. Then we consider the following system of evolutionary PDEs:
\begin{gather*}
\frac{\d v^\alpha}{\d t^\beta_q}=\eta^{\alpha\mu}\frac{\delta\og_{\beta,q}}{\delta u^\mu}.
\end{gather*}
From the compatibility of the flows of the double ramification hierarchy it easily follows that this system is also compatible. Let $(v^{\str})^\alpha(x,t^*_*;\eps)$ be a unique solution that satisfies the initial condition $\left.(v^{\str})^\alpha\right|_{t^*_*=0}=\delta^{\alpha,1}\frac{x^2}{2}$. It satisfies the following equation~(see~\cite[Eq.~(5.2)]{BR14}):
\begin{gather*}
\frac{\d(v^{\str})^\alpha}{\d t^1_0}-\sum_{n\ge 0}t^\gamma_{n+1}\frac{\d(v^{\str})^\alpha}{\d t^\gamma_n}=t^\alpha_0+\delta^{\alpha,1}x.
\end{gather*}
Differentiating this equation by~$t^1_1$ and using the fact that $\frac{\d(v^{\str})^\alpha}{\d t^1_0}=\d_x(v^{\str})^\alpha=(u^\str)^\alpha$, we get
$$
\left.\left(\frac{\d(u^{str})^\alpha}{\d t^1_1}\right)\right|_{t^*_*=0}=\delta^{\alpha,1}x.
$$
Together with~\eqref{eq:tmp1} it proves~\eqref{eq:initial1}.

Let $f^\alpha_{\beta,q}:=\eta^{\alpha\mu}\d_x\frac{\delta\og_{\beta,q}}{\delta u^\mu}$. We have
\begin{multline}\label{eq:system}
\frac{\d}{\d t^\beta_q}O(u^{\str})^\alpha=O f^\alpha_{\beta,q}-f^\alpha_{\beta,q}=\sum_{n\ge 0}\frac{\d f^\alpha_{\beta,q}}{\d u^\gamma_n}O\d_x^n(u^{\str})^\gamma-\eps\frac{\d f^\alpha_{\beta,q}}{\d\eps}-f^\alpha_{\beta,q}=\\
=\sum_{n\ge 0}\frac{\d f^\alpha_{\beta,q}}{\d u^\gamma_n}\left(\d_x^n O(u^{\str})^\gamma+n\d_x^n (u^{\str})^\gamma\right)-\eps\frac{\d f^\alpha_{\beta,q}}{\d\eps}-f^\alpha_{\beta,q}.
\end{multline}
Since $f^\alpha_{\beta,q}\in\hcA_N^{[1]}$, we have 
$$
\sum_{n\ge 0}nu^\gamma_n\frac{\d f^\alpha_{\beta,q}}{\d u^\gamma_n}-\eps\frac{\d f^\alpha_{\beta,q}}{\d\eps}-f^\alpha_{\beta,q}=0.
$$
Therefore, from~\eqref{eq:system} we obtain
$$
\frac{\d}{\d t^\beta_q}O(u^{\str})^\alpha=\sum_{n\ge 0}\frac{\d f^\alpha_{\beta,q}}{\d u^\gamma_n}\d_x^n O(u^{\str})^\gamma,\qquad 1\le\alpha,\beta\le N,\quad q\ge 0.
$$ 
This system can be considered as a system of evolutionary partial differential equations for the power series $O(u^{\str})^\alpha$. Since the initial condition~\eqref{eq:initial1} is zero, we get $O(u^{\str})^\alpha=0$.  The proposition is proved.  
\end{proof}

\subsection{Dilaton equation and the reconstruction of the hierarchy}\label{subsection:reconstruction}

Suppose we have an arbitrary cohomological field theory in genus~$0$: $c_{0,n}\colon V^{\otimes n}\to H^{\even}(\oM_{0,n};\mbC)$, with a phase space~$V$ of dimension~$N$. Let $F_0(t^*_*)$ be its potential:
$$
F_0(t^*_*):=\sum_{n\ge 3}\frac{1}{n!}\sum_{d_1,\ldots,d_n\ge 0}\left(\int_{\oM_{0,n}}c_{0,n}\left(\bigotimes^n_{i=1}e_{\alpha_i}\right)\prod_{i=1}^n\psi_i^{d_i}\right)\prod_{i=1}^nt^{\alpha_i}_{d_i}.
$$
Let
$$
\Omega^{[0]}_{\alpha,p;\beta,q}(u):=\left.\frac{\d^2 F_0}{\d t^\alpha_p\d t^\beta_q}\right|_{\substack{t^*_{\ge 1}=0\\t^\gamma_0=u^\gamma}}\in\cA_{u^1,\ldots,u^N}.
$$
Consider the genus-zero Dubrovin--Zhang hierarchy associated to our cohomological field theory. It is also called the principal hierarchy. Recall that it is the hamiltonian hierarchy defined by the sequence of local functionals
$$
\oh^{[0]}_{\alpha,p}:=\int\Omega^{[0]}_{\alpha,p+1;1,0}dx,\qquad  1\le \alpha\le N,\quad p\ge 0,
$$ 
and the hamiltonian operator $\eta\d_x$. 

Consider now an arbitrary sequence of local functionals $\oh_{\alpha,p}\in\hLambda^{[0]}_N, 1\le\alpha\le N, p\ge 0$, such that
\begin{align*}
&\{\oh_{\alpha,p},\oh_{\beta,q}\}_{\eta\d_x}=0,\\
&\left.\oh_{\alpha,p}\right|_{\eps=0}=\oh^{[0]}_{\alpha,p},\\
&\oh_{1,0}=\oh^{[0]}_{1,0}.
\end{align*}
The local functionals~$\oh_{\alpha,p}$ and the hamiltonian operator $\eta\d_x$ define a hamiltonian hierarchy of PDEs that can be considered as a deformation of the principal hierarchy. The equations of this hierarchy are
\begin{gather}\label{eq:deformed hierarchy}
\frac{\d u^\alpha}{\d t^\beta_q}=\eta^{\alpha\mu}\d_x\frac{\delta\oh_{\beta,q}}{\delta u^\mu}.
\end{gather}
Let $(\spec)^\alpha(x,t^*_*;\eps)$ be a unique solution of the system~\eqref{eq:deformed hierarchy} specified by the initial condition
$$
\left.(\spec)^\alpha\right|_{t^*_*=0}=\delta^{\alpha,1}x.
$$
We call this solution the special solution.

\begin{proposition}\label{proposition:reconstruction}
Suppose the special solution $(\spec)^\alpha(x,t^*_*;\eps)$ satisfies the following equations:
\begin{align}
&\frac{\d(\spec)^\alpha}{\d t^1_0}-\sum_{n\ge 0}t^\gamma_{n+1}\frac{\d(\spec)^\alpha}{\d t^\gamma_n}=\delta^{\alpha,1},\quad\text{(string equation)},\label{eq:string for special}\\
&\frac{\d(\spec)^\alpha}{\d t^1_1}-\eps\frac{\d(\spec)^\alpha}{\d\eps}-x\frac{\d(\spec)^\alpha}{\d x}-\sum_{n\ge 0}t^\gamma_n\frac{\d(\spec)^\alpha}{\d t^\gamma_n}=0,\quad\text{(dilaton equation)}.\label{eq:dilaton for special}
\end{align}
Then all Hamiltonians $\oh_{\alpha,p}$ are uniquely determined by the Hamiltonian~$\oh_{1,1}$ and the dispersionless parts~$\left.\oh_{\beta,q}\right|_{\eps=0}=\oh^{[0]}_{\beta,q}$.
\end{proposition}
\begin{proof}
First of all, let us recall several properties of the functions~$\Omega^{[0]}_{\alpha,p;\beta,q}$ (see e.g.~\cite{BPS12a,BPS12b}):
\begin{align}
&\frac{\d\Omega^{[0]}_{\alpha,p+1;1,0}}{\d u^\beta}=\Omega^{[0]}_{\alpha,p;\beta,0},\quad p\ge 0,\notag\\
&\frac{\d\Omega^{[0]}_{\alpha,p+1;\beta,q}}{\d u^\gamma}=\Omega^{[0]}_{\alpha,p;\mu,0}\eta^{\mu\nu}\frac{\d\Omega^{[0]}_{\nu,0;\beta,q}}{\d u^\gamma},\quad p\ge 0.\label{eq:TRR-0}
\end{align}
Therefore, we have
$$
\eta^{\alpha\mu}\d_x\frac{\delta\oh_{\alpha,p}^{[0]}}{\delta u^\mu}=\eta^{\alpha\mu}\d_x\frac{\d\Omega_{\alpha,p+1;1,0}^{[0]}}{\d u^\mu}=\eta^{\alpha\mu}\frac{\d\Omega^{[0]}_{\alpha,p;\mu,0}}{\d u^\gamma}u^\gamma_x.
$$
Since the integral $\int_{\oM_{0,3}}\psi_1^p c_{0,3}(e_\alpha\otimes e_\mu\otimes e_\gamma)$ obviously vanishes when $p\ge 1$, we get
\begin{gather}\label{eq:vanishing}
\left.\frac{\d\Omega^{[0]}_{\alpha,p;\mu,0}}{\d u^\gamma}\right|_{u^*=0}=0,\quad\text{if $p\ge 1$}.
\end{gather}

Let us prove now that the special solution~$(\spec)^\alpha$ is uniquely determined by the Hamiltonian~$\oh_{1,1}$ and the Hamiltonians~$\oh^{[0]}_{\beta,q}$ of the principal hierarchy. Since, 
$$
\oh_{1,0}=\oh^{[0]}_{1,0}=\int\frac{1}{2}\eta_{\alpha\beta}u^\alpha u^\beta dx,
$$
we have $\d_x(\spec)^\alpha=\frac{\d(\spec)^\alpha}{\d t^1_0}$. Therefore, it is enough to determine only the coefficients of~$t^{\alpha_1}_{d_1}\ldots t^{\alpha_n}_{d_n}\eps^i$ in~$(\spec)^\alpha$. We will denote these coefficients by $c_{\alpha_1,\ldots,\alpha_n;i}^{d_1,\ldots,d_n;\alpha}$. The principal hierarchy determines the coefficients~$c_{\alpha_1,\ldots,\alpha_n;0}^{d_1,\ldots,d_n;\alpha}$. Let 
$$
\eta^{\alpha\mu}\d_x\frac{\delta\oh_{1,1}}{\delta u^\mu}=\sum_{i\ge 0}P^\alpha_i(u,u_x,\ldots)\eps^i,\quad P^\alpha_i\in\cA_N, \quad\deg P^\alpha_i=i+1.
$$
The dilaton equation~\eqref{eq:dilaton for special} and equation~\eqref{eq:deformed hierarchy} for $\beta=1$ and $q=1$ imply that
\begin{gather}\label{main recursion}
\eps\frac{\d(\spec)^\alpha}{\d\eps}+x\frac{\d(\spec)^\alpha}{\d x}+\sum_{n\ge 0}t^\gamma_n\frac{\d(\spec)^\alpha}{\d t^\gamma_n}=\sum_{i\ge 0}P^\alpha_i(\spec,\spec_x,\ldots)\eps^i.
\end{gather}

Let us prove that this equation allows to compute all the coefficients~$c_{\alpha_1,\ldots,\alpha_n;i}^{d_1,\ldots,d_n;\alpha}$ starting from the coefficients~$c_{\beta_1,\ldots,\beta_m;0}^{l_1,\ldots,l_m;\gamma}$. To be precise, we are going to prove that equation~\eqref{main recursion} allows to express any coefficient~$c_{\alpha_1,\ldots,\alpha_n;i}^{d_1,\ldots,d_n;\alpha}$, $i>0$, in terms of the coefficients 
\begin{gather}
c_{\beta_1,\ldots,\beta_m;j}^{k_1,\ldots,k_m;\gamma},\label{previous coefficients}
\end{gather}
where one of the following two conditions holds:
\begin{align}
&\text{1. } j<i,\label{condition1}\\
&\text{2. } j\le i \text{ and } m<n.\label{condition2}
\end{align}
The coefficient of~$t^{\alpha_1}_{d_1}\ldots t^{\alpha_n}_{d_n}\eps^i$ on the left-hand side of~\eqref{main recursion} is equal to~$(i+n)c_{\alpha_1,\ldots,\alpha_n;i}^{d_1,\ldots,d_n;\alpha}$. Let us look at the coefficient of~$t^{\alpha_1}_{d_1}\ldots t^{\alpha_n}_{d_n}\eps^i$ on the right-hand side of~\eqref{main recursion}. The string equation~\eqref{eq:string for special} implies that
$$
\left.\d_x^d(\spec)^\alpha\right|_{\substack{x=0\\t^*_*=0}}=\delta^{\alpha,1}\delta_{d,1}
$$
and that the coefficient of $t^{\beta_1}_{l_1}\ldots t^{\beta_m}_{l_m}\eps^j, m\ge 1$, in~$\d_x^d(\spec)^\gamma$ is a linear combination of the coefficients~$c_{\beta_1,\ldots,\beta_m;j}^{k_1,\ldots,k_m;\gamma}$. Therefore, the coefficient of~$t^{\alpha_1}_{d_1}\ldots t^{\alpha_n}_{d_n}\eps^i$ in 
$$
\sum_{r\ge 1}P^{\alpha}_r(\spec,\spec_x,\ldots)\eps^r
$$ 
can be expressed in terms of coefficients~\eqref{previous coefficients} with condition~\eqref{condition1}.

We have 
$$
P^\alpha_0=\eta^{\alpha\mu}\frac{\d\Omega^{[0]}_{1,1;\mu,0}}{\d u^\gamma}u^\gamma_x.
$$
From~\eqref{eq:vanishing} it follows that the coefficient of~$t^{\alpha_1}_{d_1}\ldots t^{\alpha_n}_{d_n}\eps^i$ in
$$
\eta^{\alpha\mu}\frac{\d\Omega^{[0]}_{1,1;\mu,0}}{\d u^\gamma}((\spec)_x^\gamma-\delta^{\gamma,1})
$$
can be expressed in terms of coefficients~\eqref{previous coefficients} with condition~\eqref{condition2}. Finally, we compute
$$
\eta^{\alpha\mu}\frac{\d\Omega^{[0]}_{1,1;\mu,0}}{\d u^1}\stackrel{\text{by \eqref{eq:TRR-0}}}{=}\eta^{\alpha\mu}\underbrace{\Omega^{[0]}_{1,0;\nu,0}}_{=\eta_{\nu\theta}u^\theta}\eta^{\nu\rho}\underbrace{\frac{\d\Omega^{[0]}_{\rho,0;\mu,0}}{\d u^1}}_{=\eta_{\rho\mu}}=u^\alpha.
$$
We see that the coefficient of~$t^{\alpha_1}_{d_1}\ldots t^{\alpha_n}_{d_n}\eps^i$ of the left-hand side of~\eqref{main recursion} is equal to~$(i+n)c_{\alpha_1,\ldots,\alpha_n;i}^{d_1,\ldots,d_n;\alpha}$, while the coefficient of this monomial on the right-hand side of~\eqref{main recursion} is equal to $1$ plus a combination of coefficients~\eqref{previous coefficients} with condition~\eqref{condition1} or~\eqref{condition2}. We conclude that the special solution~$(\spec)^\alpha$ is uniquely determined by the Hamiltonians~$\oh_{1,1}$ and~$\oh^{[0]}_{\alpha,p}$. 

Since the dispersionless part $\left.\oh_{\alpha,p}\right|_{\eps=0}$ coincides with the Hamiltonian $\oh^{[0]}_{\alpha,p}$ of the principal hierarchy, we have (see e.g.~\cite{BPS12a})
$$
\left.(\spec)^\alpha\right|_{\eps=0}=\eta^{\alpha\mu}\left.\frac{\d^2 F_0}{\d t^1_0\d t^\mu_0}\right|_{t^1_0\mapsto t^1_0+x}.
$$
Then from the string equation for the potential~$F_0$ it follows that
$$
\left.\d_x^d(\spec)^\alpha\right|_{x=0}=t^\alpha_d+\delta^{\alpha,1}\delta_{d,1}+O(t^2)+O(\eps).
$$
This equation implies that any power series in the variables~$t^\nu_i$ and~$\eps$ can be written as a power series in~$\left(\left.\d_x^d(\spec)^\alpha\right|_{x=0}-\delta^{\alpha,1}\delta_{d,1}\right)$ and $\eps$ in a unique way. Thus, the special solution determines the differential polynomials~$\eta^{\alpha\mu}\d_x\frac{\delta\oh_{\alpha,p}}{\delta u^\mu}$. Since $\left.\frac{\delta\oh_{\alpha,p}}{\delta u^\mu}\right|_{u^*_*=0}=0$, the Hamiltonians~$\oh_{\alpha,p}$ are also uniquely determined. The proposition is proved.
\end{proof}

Let us show now how to use this proposition in order to relate the Dubrovin--Zhang hierarchy to the double ramification hierarchy. Consider an arbitrary semisimple cohomological field theory and the associated Dubrovin--Zhang and the double ramification hierarchies. We denote by~$w^\alpha$ the dependant variables of the Dubrovin--Zhang hierarchy, by~$\oh_{\alpha,p}$ the Hamiltonians and by~$K$ the hamiltonian operator. Consider some Miura transformation $u^\alpha\mapsto w^\alpha(u;u_x,\ldots;\eps)$.
\begin{proposition}\label{proposition:DR and DZ}
Suppose that the Miura transformation~$u^\alpha\mapsto w^\alpha(u;u_x,\ldots;\eps)$ satisfies the following conditions:
\begin{enumerate}
\item $\frac{\d w^\alpha}{\d u^1}=\delta^{\alpha,1}$;

\item The Miura transform of the standard operator $\eta\d_x$ coincides with the operator~$K$;

\item The Miura transform of the Hamiltonian $\og_{1,1}$ coincides with the Hamiltonian $\oh_{1,1}$: $\og_{1,1}[w]=\oh_{1,1}$.
\end{enumerate}
Then the Miura transform of the double ramification hierarchy coincides with the Dubrovin--Zhang hierarchy.
\end{proposition}
\begin{proof}
The Dubrovin--Zhang hierarchy has the so-called topological solution that is defined by (see~\cite{DZ05})
$$
(\wtop)^\alpha(x,t^*_*;\eps):=\eta^{\alpha\mu}\left.\frac{\d^2 F}{\d t^1_0\d t^\mu_0}\right|_{t^1_0\mapsto t^1_0+x},
$$
where $F(t^*_*;\eps)$ is the potential of the cohomological field theory. From the string and the dilaton equations for~$F$ it is easy to see that the topological solution satisfies the string and the dilaton equations~\eqref{eq:string for special} and~\eqref{eq:dilaton for special}, and also the initial condition
$$
\left.(\wtop)^\alpha\right|_{t^*_*=0}=\delta^{\alpha,1}x.
$$
Consider the inverse Miura transformation~$w^\alpha\mapsto u^\alpha(w;w_x,\ldots;\eps)$. Let 
$$
(u^\mathrm{top})^\alpha(x,t^*_*;\eps):=\left.u^\alpha\right|_{w^\gamma_n=\d_x^n(\wtop)^\gamma}.
$$
It is easy to see that the power series~$(u^\mathrm{top})^\alpha$ satisfies the dilaton equation~\eqref{eq:dilaton for special}. Note that condition (1) is equivalent to the condition $\frac{\d u^\alpha}{\d w^1}=\delta^{\alpha,1}$, which easily implies that the power series~$(u^\mathrm{top})^\alpha$ satisfies the string equation~\eqref{eq:string for special}. From condition~(2) it follows that the inverse Miura transform of the Dubrovin--Zhang hierarchy satisfies the assumptions of Proposition~\ref{proposition:reconstruction}.

On the other hand, by~\cite[Lemma 4.4]{Bur14} the genus-zero part of the double ramification hierarchy coincides with the genus-zero part of the Dubrovin--Zhang hierarchy, and, by \cite[Lemma 4.7]{Bur14} and Proposition~\ref{proposition:dilaton} the string solution~$(u^{\str})^\alpha$ satisfies the string and the dilaton equations~\eqref{eq:string for special} and~\eqref{eq:dilaton for special}. Therefore, the double ramification hierarchy also satisfies the assumptions of Proposition~\ref{proposition:reconstruction}. As a result, condition~(3) and Proposition~\ref{proposition:reconstruction} imply that the Dubrovin--Zhang and the double ramification hierarchy are equivalent by our Miura transformation.
\end{proof}

\subsection{Proof of Theorem~\ref{theorem:main}}\label{subsection:final proof}

We just check all the conditions from Proposition~\ref{proposition:DR and DZ}. Condition~(1) is obvious. For condition (2) we use the formulas for the operator~$K^{\text{$r$-spin}}$ from Section~\ref{subsection:examples}. In order to check condition~(3) we use our computations of the Hamiltonian~$\og^{\text{$r$-spin}}_{1,1}$ from~Propositions~\ref{proposition:DR 3-spin},~\ref{proposition:DR 4-spin} and~\ref{proposition:DR 5-spin} and the formulas for the Hamiltonian~$\oh^{\text{$r$-spin}}_{1,1}$ from Section~\ref{subsection:examples}. Then the theorem follows from Proposition~\ref{proposition:DR and DZ}.


\section{Quantization of the $r$-spin Dubrovin--Zhang hierarchy for $r=3,4,5$}\label{section:quantization}

A quantization of the $2$-spin Dubrovin--Zhang hierarchy was constructed in~\cite{BR15}. In this section we obtain a quantization of the $r$-spin Dubrovin--Zhang hierarchy for $r=3,4,5$. This is a consequence of Theorem~\ref{theorem:main} and the construction of~\cite{BR15}.

\subsection{Quantization of the double ramification hierarchy}

Consider an arbitrary cohomological field theory $c_{g,n}\colon V^{\otimes n}\to H^\even(\oM_{g,n};\mbC)$.
In~\cite{BR15} a natural quantization of the associated double ramification hierarchy was constructed. Let us briefly recall it. 

First of all, we have to introduce the Weyl algebra~$\mathfrak{W}_N$. It is formed by (power series in $\hbar$ with coefficients that are) power series in~$p^\alpha_k$, $k\leq 0$, with coefficients that are polynomials in~$p^\alpha_k$, $k>0$, with $\alpha=1,\ldots,N$. The product rule is described as follows: representing two power series in the ``normal form'', i.e. with all variables with negative or zero subscripts appearing on the left of all variables with positive subscripts,
$$
f=\sum_{g\geq 0} \sum_{n\geq 0}\ \sum_{k_1,\ldots, k_n \leq 0} p^{\alpha_1}_{k_1}\ldots p^{\alpha_n}_{k_n} f^{\alpha_1,\ldots,\alpha_n}_{k_1,\ldots,k_n;g}(p_{k>0})\hbar^g,
$$
$$
g=\sum_{g \geq 0} \sum_{n\geq 0}\ \sum_{k_1,\ldots, k_n\leq 0} p^{\alpha_1}_{k_1}\ldots p^{\alpha_n}_{k_n} g^{\alpha_1,\ldots,\alpha_n}_{k_1,\ldots,k_n;g}(p_{k>0}) \hbar^g,
$$
where $f^{\alpha_1,\ldots,\alpha_n}_{k_1,\ldots,k_n;g}(p_{>0})$ and $g^{\alpha_1,\ldots,\alpha_n}_{k_1,\ldots,k_n;g}(p_{>0})$ are polynomials, one obtains the product $f\star g$ by commuting the $p_{\leq 0}$ variables of $g$ with the $p_{k>0}$ variables of $f$ using $[p^\alpha_k,p^\beta_j]=i \hbar k \eta^{\alpha \beta}\delta_{k+j,0}$. Thanks to polynomiality of the coefficients, this process is well-defined and produces another element of the same Weyl algebra~$\mathfrak{W}_N$.

For $1\le\alpha\le N$ and $d\ge 0$ define the following elements of the algebra~$\fW_N[[\eps]]$:
\begin{gather}\label{eq:quantum Hamiltonians}
\overline G_{\alpha,d}:=\sum_{\substack{g\ge 0,n\ge 0\\2g-1+n>0}}\frac{(i \hbar)^g}{n!}\sum_{a_1+\ldots+a_n=0}\left(\int_{\DR_g\left(0,a_1,\ldots,a_n\right)}\Lambda\left(\frac{-\eps^2}{i \hbar}\right) \psi_1^d c_{g,n+1}\left(e_\alpha\otimes \bigotimes_{i=1}^n e_{\alpha_i}\right)\right)p^{\alpha_1}_{a_1}\ldots p^{\alpha_n}_{a_n},
\end{gather}
where $\Lambda\left(\frac{-\eps^2}{i \hbar}\right):=1+\left( \frac{-\eps^2}{i \hbar}\right) \lambda_1+\ldots + \left(\frac{-\epsilon^2}{i\hbar}\right)^g \lambda_g$. It is easy to see that
$$
\left.\overline G_{\alpha,d}\right|_{\hbar=0}=\og_{\alpha,d}.
$$
In~\cite{BR15} it is proved that the elements $\overline G_{\alpha,d}$ mutually commute. Therefore, the construction~\eqref{eq:quantum Hamiltonians} gives a quantization of the double ramification hierarchy.

\subsection{$r$-spin theory for $r=3,4,5$}

Consider now the $r$-spin theory. In the $3$-spin case, by Theorem~\ref{theorem:main}, the Dubrovin--Zhang hierarchy coincides with the double ramification hierarchy. Therefore, the construction of~\cite{BR15} immediately gives a quantization of the $3$-spin Dubrovin--Zhang hierarchy.

Suppose now that $r$ is equal to~$4$ or~$5$. Quantization of the Dubrovin--Zhang hierarchy in these cases is slightly different because the hamiltonian structure doesn't coincide with the standard one. However, this is easy to handle. Note that the hamiltonian operator $K^{\text{$r$-spin}}=((K^{\text{$r$-spin}})^{\alpha\beta})$ has the form
\begin{gather}\label{eq:good form}
(K^{\text{$r$-spin}})^{\alpha\beta}=\sum_{i\ge 0}(K^{\text{$r$-spin}})^{\alpha\beta}_i\eps^i\d_x^{i+1},
\end{gather}
where $(K^{\text{$r$-spin}})^{\alpha\beta}_i$ are constants. It is very easy to quantize the Poisson structure corresponding to this operator. Let $\widetilde p^\alpha_n$ be the Fourier components of the fields $w^\alpha(x)$:
$$
w^\alpha(x)=\sum_{n\in\mbZ}\widetilde p^\alpha_n e^{inx}.
$$
Introduce a deformed algebra $\widetilde\fW_N[[\eps]]_{K^{\text{$r$-spin}}}$ as follows. As a vector space it coincides with the space~$\fW_N[[\eps]]$, but with the variables $p^\alpha_n$ replaced by $\widetilde p^\alpha_n$. We endow the space~$\widetilde\fW_N[[\eps]]_{K^{\text{$r$-spin}}}$ with a product rule using the following deformed commutation relation:
$$
[\widetilde p^\alpha_m,\widetilde p^\beta_n]_{K^{\text{$r$-spin}}}:=\hbar\delta_{m+n,0}\sum_{j\ge 0}\eps^j(im)^{j+1}(K^{\text{$r$-spin}})^{\alpha\beta}_j.
$$
It is clear that this gives a quantization of the Poisson structure on~$\hLambda_{w^1,\ldots,w^{r-1}}^{[0]}$ defined by the operator~$K^{\text{$r$-spin}}$. The Miura transformation from Theorem~\ref{theorem:main} induces an isomorphism $f_r\colon\widetilde\fW_N[[\eps]]_{K^{\text{$r$-spin}}}\to\fW_N[[\eps]]$ that is given by
$$
r=4\colon
\left\{
\begin{aligned}
f_4(\tp^1_n)=&p^1_n-\frac{\eps^2}{96}n^2p^3_n,\\
f_4(\tp^2_n)=&p^2_n,\\
f_4(\tp^3_n)=&p^3_n;
\end{aligned}
\right.
\qquad\qquad
r=5\colon
\left\{
\begin{aligned}
f_5(\tp^1_n)=&p^1_n-\frac{\eps^2}{60}n^2p^3_n,\\
f_5(\tp^2_n)=&p^2_n-\frac{\eps^2}{60}n^2 p^4_n,\\
f_5(\tp^3_n)=&p^3_n,\\
f_5(\tp^4_n)=&p^4_n.
\end{aligned}
\right.
$$
We see that that for the $4$ and $5$-spin theories the elements $f_r^{-1}(\overline G_{\alpha,d})\in\widetilde\fW[[\eps]]_{K^{\text{$r$-spin}}}$ define a quantization of the Dubrovin--Zhang hierarchy.



\bibliographystyle{plain} 

\bibliography{bibliographytowardadescription} 


\end{document}